\documentclass[11pt,onecolumn,journal,draftcls,oneside]{IEEEtran}

\ifCLASSINFOpdf
   \usepackage[pdftex]{graphicx}
\else
   \usepackage[dvips]{graphicx}
\fi

\usepackage{graphicx}
\usepackage{url}
\usepackage{subfigure}
\usepackage{relsize}
\usepackage{amsmath}
\usepackage{amssymb}
\usepackage{amsthm}
\usepackage{amsfonts}
\usepackage{paralist}
\usepackage{multirow}
\usepackage{booktabs}
\usepackage{bm}
\usepackage{algorithm}
\usepackage{algorithmic}
\usepackage{authblk}
\usepackage{xcolor}
\usepackage{cancel}
\usepackage{subfigure}

\newcommand{\ds}{\displaystyle}

\newtheorem{theorem}{Theorem}
\newtheorem{lemma}[theorem]{Lemma}

\newtheorem{corollary}[theorem]{Corollary}

\newcounter{cond}
\newtheorem{condition}[cond]{Condition}

\newcounter{rema}
\newtheorem{remark}[rema]{Remark}

\newcounter{exam}
\newtheorem{example}[exam]{Example}

\DeclareFontFamily{OT1}{pzc}{}
\DeclareFontShape{OT1}{pzc}{m}{it}%
              {<-> s * [1.3] pzcmi7t}{}
\DeclareMathAlphabet{\mathpzc}{OT1}{pzc}%
                                 {m}{it}

\newcommand{\CN}{\mathcal{N}}
\newcommand{\CH}{\mathcal{H}}
\newcommand{\Cc}{\mathpzc{c}}
\newcommand{\CV}{\mathcal{V}}
\newcommand{\CE}{\mathcal{E}}
\newcommand{\Cm}{\mathpzc{m}}
\newcommand{\CG}{\mathcal{G}}

\newcommand{\Bd}{\boldsymbol{d}}

\newcommand{\Br}{\boldsymbol{r}}

\newcommand{\Ba}{\boldsymbol{a}}
\newcommand{\BG}{\boldsymbol{G}}
\newcommand{\Be}{\boldsymbol{e}}

\newcommand{\Bpi}{\boldsymbol{\pi}}
\newcommand{\BB}{\boldsymbol{B}}
\newcommand{\BT}{\boldsymbol{T}}
\newcommand{\Blambda}{\boldsymbol{\lambda}}

\begin{document}

\title{{\huge Minimum-Length Scheduling with Finite Queues:\\
Solution Characterization and Algorithmic Framework}
\thanks{\em Part of the material of this paper was presented in ISIT 2012 -- IEEE International Symposium on Information Theory.}%
}


\author[1]{Vangelis Angelakis}
\author[1,2]{Anthony Ephremides}
\author[1]{Qing He}
\author[1]{Di Yuan}
\affil[1]{Department of Science and Technology, Link{\"o}ping University, Sweden}
\affil[2]{Department of Electrical and Computer Engineering,
University of Maryland, USA}

\markboth{Submitted to IEEE Transactions on Information Theory, July 2012.}%
{Submitted to IEEE Transactions on Information Theory, July 2012.}

\maketitle

\begin{abstract}

We consider a set of transmitter-receiver pairs, or links, that share a common channel and address the problem of emptying backlogged queues at the transmitters in minimum time.
The problem amounts to determining activation subsets of links and their time durations to form a minimum-length schedule.
The problem of scheduling has been studied under various formulations before.
In this paper, we present fundamental insights and solution characterizations that include: (i) showing that the complexity of the problem remains high for any continuous and increasing rate function, (ii) formulating and proving sufficient and necessary optimality conditions of two base scheduling strategies that correspond to emptying the queues using ``{\em one-at-a-time}'' or ``{\em all-at-once}'' strategies, (iii) presenting and proving the tractability of the special case in which the transmission rates are functions only of the cardinality of the link activation sets.
These results are independent of physical-layer system specifications and are valid for any form of rate
function.
We then develop an algorithmic framework.
The framework encompasses exact as well as sub-optimal, but fast, scheduling algorithms, all under a unified principle design. Through computational experiments we finally investigate the performance of several specific algorithms.

\end{abstract}


\section{Introduction}
\label{sec:introduction}

For multiple communication links sharing a common wireless channel, the fundamental aspect of access coordination is called scheduling.
It amounts to deciding which links are allowed to transmit simultaneously and for how long they should do so.
Usually, the selection of a schedule is driven by the goal of optimizing a cost criterion.
Scheduling has a long history of investigation that has ranged from simple transmission models to fully cross-layered ones that combine rate and power control with overall network resource allocation.
In this paper, we examine a version of the scheduling problem that arises from the objective of draining in minimum time the bit-contents that reside at the transmitters of a finite number of links.
That is, we consider the multiple access or interference channel with finite traffic volume that must be delivered in minimum time.

Past work on this problem includes \cite{HaSa88} in which a centralized, polynomial-time algorithm was obtained for static networks with specified link traffic requirements.
The formulation was based on mapping the network to an undirected graph and on assuming that any two links can be successfully activated simultaneously as long as they do not share common vertices on the graph.
In \cite{BoLiXi10, GoPsWa07} ${\cal NP}$-hardness was addressed for the problem of determining a minimum-length schedule under a given traffic demand in a wireless network with SINR constraints, using a protocol and  a geometric model respectively.
In some special cases the structure of the traffic demand allowed a polynomial algorithm \cite{BoEp06}.

In \cite{BjVaYu03, BjVaYu04} it was shown that more fundamental resource allocation problems in wireless networks with SINR constraints, such as node and link assignment, are also ${\cal NP}$-hard.
In these problems the goal is to assign at least one time slot to each node, or link, such that the number of time slots is minimized.
Set-covering formulations enabled a column generation method for solving the resulting linear programming relaxations.
For the minimum length scheduling problem, a column-generation-based solution method was also used in \cite{KoWi10}, which can approach an optimal solution, with the advantage of a potentially reduced complexity.
In \cite{PaEp08} the minimum-length scheduling problem was formulated as a shortest path problem on directed acyclic graphs and the authors obtained suboptimal analytic characterizations.
It is also possible to ``absorb'' the scheduling in the general network resource allocation problem as done in \cite{GeNeTa06}, where the overall criterion is to maintain network stability. However basic versions of scheduling remain important, both from the theoretical standpoint and from that of specific applications.

Our contributions include new results on the combinatorial complexity of the problem, on the structure of the optimal schedule, and finally, new necessary and/or sufficient conditions for optimality based on the values of the transmission rates that the links can transmit at, when these rates depend explicitly or implicitly on the set of links who are allowed to transmit simultaneously. Thus, our results contribute to the tightening of the joint use of physical and MAC layer approaches and point to practical and realistic algorithms for approximating, or precisely determining, an optimal schedule. To that effect we also provide a number of algorithms the performance of which we evaluate extensively.

\section{System Model}
\label{sec:system}

We consider a set $\CN = \{1, \dots, N\}$ of links, or source-destination pairs that share a common channel.
These  links are associated with a strictly positive vector of demand $\Bd = \{d_1, \dots, d_N\}^T$, with each $d_i$ representing the amount of bit-traffic stored at the transmitter of the corresponding link $i$.
Without loss of generality, assume that the entries in the demand vector are in ascending order and that they take values in a continuum.
Let $\CH$ denote the union of all subsets of $\CN$, excluding the empty set.
Clearly, $|\CH| = 2^N-1$.
We use the term {\em group} to refer to a member $\Cc \in \CH$, that is, a subset of the link set.
Scheduling a group $\Cc$ means that all members of $\Cc$ are activated simultaneously for a positive amount of time.
For any group, the service rate of each of its members is a function of the group composition.
Let $F$ denote the rate function; that is, for $\Cc \in \CH$ and link $i \in \CN$, $r_{i\Cc}=F(i, \Cc)$ represents the non-negative service rate of link $i$, if $\Cc$ is active.
Clearly, the rate values can be positive only for the members of $\Cc$, i.e., $r_{i\Cc} = 0, i \notin \Cc$. If $\Cc$ is a singleton link $i$, we use $r_{ii}$ as a more convenient, short-hand notation for the rate instead of $r_{i, \{i\}}$.

In all applications with meaningful physical interpretations, the service rate has the following property:
If two elements are served together, the rates of being served can not be higher than the individual rates, respectively.
Thus, throughout the paper, it is assumed that the service rate of any link in a group does not increase if the group is augmented, i.e., for any two groups $\Cc_1 \subset \Cc_2$ and $i \in \Cc_1 \cap \Cc_2$, $F(i, \Cc_1) \geq F(i, \Cc_2)$.
We refer to this as the {\em rate monotonicity property}.
No further conditions are imposed on $F$.

The minimum-length scheduling problem amounts to, given $(N, \Bd, F)$, selecting a set of groups $\Cc_1 \dots \Cc_k$, among the $2^N-1$ members of $\CH$, along with their respective activation durations $T_j, j = 1\dots k$, so that $\sum_{j=1}^{k}T_j$ is minimized, subject to the requirement that all stored traffic is successfully delivered.
It is important to stress that the problem input does not include the explicit knowledge of the $2^N-1$ rate vectors.
If these vectors are all computed a priori, solving the problem reduces to optimizing a linear program (of which the size is exponential in $N$).
What is provided instead in the problem input is the function $F$, that can be viewed as a black box, or an ``oracle'', that returns the rate values for any given $\Cc \in \CH$.
Thus a scheduling algorithm is regarded of exponential complexity,  if in the algorithm the number of times that function $F$ is invoked is exponential in $N$.
We assume that the computation of $F$ is practically efficient, that is, one function evaluation $F(i, \Cc)$ of any group $\Cc \in \CH$ and $i \in \Cc$ runs in polynomial time in $N$.
Note that from a communication/information-theoretic perspective the rate values represent any feasible, or achievable, rates for a given channel with specific coding, modulation and detection structures.
Thus the treatment of the problem is decoupled from the physical-layer aspects of it, although it is directly connected to, and dependent on, them.

One specific scenario of interest is the highly symmetric case in which the rate is determined completely by group cardinality, and hence all group members share the same rate.
That is, $F$ is a function of $|\Cc|$ but not of its individual members.
It corresponds to a system where all receivers are located at a central point, with $N$ transmitters having the same distance (on a circle) to the center with identical geometric channel gain.
Such a case is considered in \cite{BoEp06}.
This special case is much more structured, and it is possible to derive strong results on tractability and optimality characterization of the optimal solution.
When the rates depend only on group size, the input can be equivalently defined using an $N$-dimensional rate vector $\Br = (r_1, \dots, r_N)$, each denoting the common rate of every link in a group of size $1 \dots N$ respectively.
Rate monotonicity then implies that $r_1 \geq \dots \geq r_N$.
We will subsequently use the input triplet $(N, \Bd, \Br)$ to refer to this problem case.


\section{The Rate Function $F$}

Thus far, the minimum-length scheduling problem has been presented in a rather generic form, that is the function $F$ could be completely arbitrary, provided it satisfies the monotonicity property.
If we assume that transmission is successful at some given rate on a link it means that for a fairly broad class of channel models and receiver structures that the signal-to-interference-plus-noise ratio (SINR) at the receiver must exceed a certain threshold \cite{Andrea05}.
Specifically, if a channel matrix $\BG$ of dimension $N \times N$, is provided, where its element $G_{ij}$ is the channel gain between the transmitter of link $i$ and the receiver of link $j$ and if $P_i$ denotes the power of link $i$, and $\sigma^2$ the noise variance, then for link $i$ in group $\Cc$ the SINR is given by

\begin{equation}
\label{eq:sinr}
\gamma_{i\Cc} = {\frac{P_iG_{ii}}{\ds \sum_{k \in \Cc, k \not=i}P_kG_{ki} + \sigma^2}}~.
\end{equation}

The treatment of the scheduling problem in this paper does not depend on a {\em specific form} of the rate function $F$ and, hence, it applies to emptying $N$ backlogged queues in minimum time for {\em any} system, not limited to wireless links on a common channel. For this specific context, two commonly used modeling approaches for defining $F$ are as follows.

The first is a one-step function returning either zero (no success) or one (success) as the rate value. Indeed, many of the previous studies of scheduling in wireless networks use implicitly this function
(e.g. \cite{BjVaYu03,BjVaYu04}).
In effect, a transmission of a packet is successful if and only if the SINR meets a threshold $\gamma^*$.
A group $\Cc$ such that all of its links can successfully transmit is sometimes referred to as a {\em
 feasible matching}.
Clearly, an infeasible matching will not be part of the optimal schedule, since if it would be used, it is de facto replaced by a subset of members having the SINR condition satisfied.
An equivalent view is, in the definition of $F$, to set zero rates for {\em all} elements of any infeasible matching.
Thus the following definition of $F$ provides the SINR-threshold-based model of scheduling. In the sequel, we use $F_{\mathbb B}$ to denote this binary function.

\vspace{-8mm}
$$ \begin{array}{ll}
\\[1em] r_{i\Cc} = F_{\mathbb B}(i, \Cc) & = \left\{
\begin{array}{ll} 1 & \textrm{if $i \in \Cc$ and $\gamma_{j\Cc} \geq \gamma^*, \forall j \in \Cc$,} \\
                  0 & \textrm{otherwise.}
\end{array} \right.
\end{array}
$$

\vspace{2mm}

The definition can be further generalized to account for rate adaptation.
In this case, the rate values form a discrete set with cardinality higher than two.
Each rate value is associated with an SINR threshold, often obtained from the available adaptive
modulation and coding schemes of some specific wireless system \cite{Proakis00}.
The generalization corresponds to $F$ being a step-wise function taking multiple values.

The second commonly used modeling approach is to consider the rate as a continuous function of the SINR \cite{Andrea05}.
We will use $F_{\mathbb C}$ as a general notation of the wide class of continuous functions that are (strictly) monotonically increasing in the SINR.
A particular case of interest is the Shannon formula for the additive white gaussian noise (AWGN) channel. This case will be referred to as $F_{\mathbb S}$, and is given by

\begin{equation}
\label{eq:shannon}
r_{i\Cc} = F_{\mathbb S}(i, \Cc) = \log_2(1 + \gamma_{i\Cc}).
\end{equation}

The aforementioned property of rate monotonicity clearly holds for both $F_{\mathbb B}$ and $F_{\mathbb C}$.
For $F_{\mathbb B}$, we have $r_{i{\Cc_1}} = r_{i{\Cc_2}}$, for two groups $\Cc_1 \subset \Cc_2$ and $i \in \Cc_1 \cap \Cc_2$, if and only if both are feasible matchings or both are infeasible matchings.
If $\Cc_1$ is feasible but $\Cc_2$ is not, $1=r_{i{\Cc_1}} > r_{i{\Cc_2}} = 0, i \in \Cc_1 \cap \Cc_2$.
For $F_{\mathbb C}$, strict inequality $r_{i{\Cc_1}} > r_{i{\Cc_2}}$ holds as long as $\Cc_1 \subset \Cc_2$.


\section{Linear Programming Formulation}
\label{sec:linear}

The scheduling problem is easily shown to be equivalent to a linear program (LP).
Although formulating the LP does not give a practically feasible solution algorithm, it enables us to gain structural insights.
Denote by $\BT = T_\Cc, \Cc \in \CH$ the non-negative scheduling decision vector of dimension
$2^N-1$, whose element $T_{\Cc}$ denotes the time duration of running group $\Cc \in \CH$.
We use $T^*$ to denote an optimal scheduling solution.
Notation $\CH^*$ is reserved for a set of groups that correspond to an optimum solution, that is, $\CH^* = \{\Cc \in \CH: T_\Cc^* > 0\}$.
By the following lemma, all demands will be met exactly at optimum. This is rather intuitive and has been (implicitly) taken for granted (e.g., \cite{BoEp06}).
Formalizing this result is useful in our case, as it eliminates any doubt about the validity of the form of LP basic solutions to be discussed later.

\begin{lemma}
\label{theo:equality}
There exists an optimal schedule such that, before reaching the end of the time duration of a group, none of the link queues in the group is empty.
\end{lemma}
\begin{proof}
Suppose the opposite is true.
Then there exists a group $\Cc$ run with time duration $T_\Cc>0$ and link $i \in \Cc$, such that the demand served of $i$ in the group, denoted by $d_i$, satisfies the condition $d_i < r_{i\Cc}T_\Cc$. Let $t = \frac{d_i}{r_{i\Cc}}$.
Consider splitting the running time $T_\Cc$ in two segments, with lengths $t$ and $T-t$ respectively. In the first segment, group $\Cc$ is run, and for the second segment, the reduced group $\Cc \setminus \{i\}$ is run.
The lemma follows from two observations.
First, the served demand of $i$ in segment one remains $d_i$.
Second, any of the links other than $i$ is served for an overall time of $T$, and their rates in $\Cc \setminus \{i\}$ are not worse, if not better, than those in $\Cc$.
\end{proof}
\vspace{2mm}

By Lemma \ref{theo:equality}, we arrive at the following LP formulation.

\begin{subequations}
\label{eq:lp}
\begin{align}
\min~~ & \sum_{\Cc \in \CH} T_\Cc, \label{eq:lpobj}\\
\text{s.~t.}~~ & \sum_{\Cc \in \CH} r_{i\Cc} T_\Cc = d_i ~~i=1, \dots, N \label{eq:lpcons},\\
& \BT \geq 0.
\end{align}
\end{subequations}

As \eqref{eq:lpcons} are equalities, the formulation is in so called standard LP form, hence no slack or surplus variables will be involved in constructing matrix bases or the corresponding basic solutions.

Even though there are $2^N-1$ candidate groups, we can conclude the existence of an optimal scheduling solution using at most $N$ groups.
The result follows from the fundamental optimality theory of LP and the structure of \eqref{eq:lp}.

\begin{lemma}
\label{theo:ngroups}
There exists an optimal scheduling solution using at most $N$ groups, i.e., $|\CH^*| \leq N$.
\end{lemma}
\begin{proof}
Note that the feasible region of \eqref{eq:lp} is non-empty as the $N$ single-link groups would provide a feasible schedule.  This corresponds to a TDMA-based activation of the links one-at-a-time until each empties its queue.
Hence, by the fundamentals of linear programming (e.g., \cite{Mu83}), there exists an optimal basic solution.
For any feasible basic solution, the number of positive values is no more than the number of rows $N,$ and equals $N$ if the solution is non-degenerate, and the lemma follows.
\end{proof}

By Lemma \ref{theo:ngroups}, there is always a compact representation of optimality.
However, finding this best combination of $N$ groups, among the $2^N-1$ candidate ones, remains generally hard (see also Section \ref{sec:complexity}).
In this regard, optimal scheduling has a combinatorial side, even if formulation \eqref{eq:lp} is an LP.

In some of the analysis later on, we utilize the LP dual of \eqref{eq:lp}. Letting $\pi_i$ denote the dual variable of \eqref{eq:lpcons},
the dual formulation is as follows.

\begin{subequations}
\label{eq:dual}
\begin{align}
\max~~ & \sum_{i \in \CN} d_i \pi_i, \label{eq:dualobj}\\
\text{s.~t.}~~ & \sum_{i \in \Cc} r_{i\Cc} \pi_i \leq 1~~\Cc \in \CH, \label{eq:dualcons} \\
& \Bpi \geq 0.
\end{align}
\end{subequations}

\section{Complexity Considerations}
\label{sec:complexity}

Complexity is a fundamental aspect in the treatment of optimization problems.
By Lemma \ref{theo:ngroups}, obtaining the globally optimal schedule is equivalent to selecting the $N$ ``best'' groups.
The question is how difficult this selection task is.
For a discrete rate function $F_{\mathbb B}$, the problem is ${\cal NP}$-hard \cite{AnDi09,BjVaYu04,GoPsWa07}.
A natural follow-up question is whether the complexity reduces for continuous (and thus much more well-behaved) functions of class $F_{\mathbb C}$.
For example, is the problem tractable if the rate function is $F_{\mathbb S}$, or, even simply linear (regardless of the fact that it would not be realistic) in SINR?
In the following, we provide a negative answer, stating that the problem in the wireless communications context is in general hard for {\em all} rate functions that are continuous and strictly increasing in the SINR.

\begin{theorem}
\label{theo:hard}
Given any function $F_{\mathbb C}$ of the SINR, there are {\cal NP}-hard instances of the minimum-length scheduling problem.
\end{theorem}

\begin{proof}
Given $(N, \Bd, F)$, where $F$ is of type $F_{\mathbb C}$, the recognition version of the problem, by Lemma \ref{theo:ngroups}, is as follows.
Are there $N$ groups, which can be represented using a binary $N \times N$ matrix, such that the total time of satisfying $\Bd$ using these groups is at most a given positive number?
The problem is clearly in class ${\cal NP}$, as checking the validity of a solution (a certificate in form of a square matrix of size $N$) is straightforward 
Consider a general-topology graph $\CG = (\CV, \CE)$. Let $N = |\CV|$. Thus a link in the scheduling instance corresponds to a vertex in $\CG$.
Let $v = F^{-1}(\frac{1}{N})$, and $u = \frac{F^{-1}(1)}{F^{-1}(\frac{1}{N})}$, i.e., $F^{-1}(1) = vu$, with $u>1$ because $F$ is strictly increasing in SINR.  Let $\sigma^2 = \frac{1}{u}$. For each edge $(i,j)$ in the graph, set the coupling element $G_{ij}=G_{ji}=1$. Moreover, $G_{ii} = G = \min \{v, 1.0\}, i \in \CN$. All other elements of the channel matrix are zeros.
Finally, the transmit power $P_i = \frac{v}{G_{ii}}, i \in \CN$.

Consider link $i$ and any group that contains $i$, but not any of the  adjacent vertexes in $\CG$.
The SINR is $vu = F^{-1}(1)$, thus the rate is 1.0. If $i$ is put in a group containing at least one adjacent vertex in $\CG$, the SINR is no more than $v/(\frac{v}{G_{ii}} + \frac{1}{u}) < v = F^{-1}(\frac{1}{N})$, because $\frac{v}{G_{ii}}\geq 1$ and $u>0$.
Thus the rate of $i$ becomes strictly less than $\frac{1}{N}$.
Suppose, at optimum, a group $\Cc$ containing two links $i$ and $j$, that are adjacent in $\CG$, has a positive amount of time duration $T>0$.
Note that, in $\CG$, $\Cc$ corresponds to at least one connected component (because $i$ and $j$ are adjacent).
Denote by $\Cm \subseteq \CV$ the component containing $i$ and $j$, and let $m = |\Cm|$.
Note that $m \geq 2$.
By the observation before, for each of the links in $\Cm$, including $i$ and $j$, the demand served in time $T$ within group $\Cc$ is strictly less than $\frac{T}{N}$.

Consider splitting group $\Cc$ into $m$ groups, obtained by combining $\Cc \setminus \Cm$ with each of the individual links in $\Cm$.
Each of the $m$ groups is given time $\frac{T}{m}$.
For all links in $\Cm$, including $i$ and $j$, the rate grows from less than $\frac{1}{N}$ to $1$.
Since $m \leq N$, the quantity $\frac{T}{m}$ is strictly more than enough to serve demand $\frac{T}{N}$, for any link in $\Cm$.
For the links in $\Cc \setminus \Cm$, overall they are served with the same time duration $T$, with rate no less than before.
Repeat the argument for the remaining components if necessary.
In conclusion, there is an optimal scheduling solution in which the groups are formed by links corresponding to independent sets of $\CG$.
At this stage, it is apparent that solving the scheduling problem provides the correct answer to the weighted fractional coloring problem \cite{LuYa94}, with the demand vector $\Bd$ being the weights of the vertexes, and the result follows.
\end{proof}

Theorem \ref{theo:hard} establishes the inherent difficulty of the scheduling problem.
The result generalizes the observation made in \cite{BoLiXi10} on the connection between fractional coloring and scheduling under the so called protocol model, which uses a conflict graph and disregards the channel matrix.
As our result applies to {\em any} $F_{\mathbb C}$, one should not expect that the use of smooth rate functions, including linear ones, would help in reducing complexity.

\section{Optimality Conditions for Base Scheduling Strategies}
\label{sec:condition}

We consider two base strategies that are the most simple choices in constructing a scheduling solution.
In the first strategy, denoted by $\CH^1$, the link queues are emptied completely separately, corresponding to a TDMA activation.
That is, $\CH^1 = \{\{1\}, \{2\}, \dots, \{N\}\}$.
The second strategy, denoted by $\CH^N$,  applies the very opposite philosophy, namely, all links are activated at once, and the $N$-links group is served until some of the queues becomes empty.
The next group consists in all the links having positive remaining demand, and so on.

Note that both strategies have size $N$, and hence represent basic solutions (extreme points of the polytope in the LP formulation).
Given $N$ out of the $2^N-1$ groups, the computing time of the correct time share (or concluding that the $N$ groups do not form a feasible schedule) is normally of complexity $O(N^3)$ due to matrix inversion.
Solutions $\CH^1$ and $\CH^N$ are simpler to construct -- after $N$ calls of function $F$, computing the $\CH^1$ schedule runs clearly in linear time, whereas for the $\CH^N$ schedule the computing time is of $O(N^2)$.

Intuitively, strategy $\CH^1$ is desirable, if the links, when activated simultaneously with others, experience significant rate reduction. This corresponds to a high-interference environment.
The following condition quantifies the notion.

\begin{condition}
\label{cond:sep}
For all $\Cc \in \CH$, the sum of the ratios between the members' rates in $\Cc$ and their respective rates
of being served individually, is at most 1.0, that is,
$$
\sum_{i \in \Cc} \frac{r_{i\Cc}}{r_{ii}} \leq 1~~ \forall \Cc \in \CH.
$$
\end{condition}

The above condition is simple in structure. Yet, it is exact in characterizing the optimality of $\CH^1$.

\begin{theorem}
\label{theo:sep}
$\CH^1$ is optimal if and only if Condition \ref{cond:sep} holds.
\end{theorem}
\begin{proof}
Sufficiency: Consider the LP formulation \eqref{eq:lp}, and the base matrix $\BB$ for the basic solution $\CH^1$.
The inverse matrix $\BB^{-1}$ is diagonal with $\text{diag}(\BB^{-1}) = (1/r_{11}, \dots, 1/r_{NN})'$, where by $\Ba'$ we denote the transpose of vector $\Ba$. For any non-basic variable $T_{\Cc}$ with $|\Cc| \geq 2$, the reduced cost equals $1 - \Be' \BB^{-1} {\Br_{\Cc}}$, where $\Be'$ is a row vector of $N$ ones and $\Br_{\Cc}$ denotes the column vector corresponding to $T_{\Cc}$ in \eqref{eq:lp}.

The expression leads to value $1 - \sum_{i \in \Cc} \frac{r_{i\Cc}}{r_{ii}}$, that is non-negative if Condition \ref{cond:sep} holds.
Since none of the $2^N-N-1$ non-basic variables has strictly negative reduced cost, $\CH^1$ is optimal, by LP optimality. \\
Necessity: If Condition \ref{cond:sep} does not hold for some group $\Cc$, the reduced cost of the corresponding non-basic variable is strictly negative.
Moreover, for $\CH^1$, all the basic variables have strictly positive values.
Therefore the LP pivot operation of bringing in $T_{\Cc}$ into the base is not degenerate, meaning that the objective function will strictly improve, and the result follows.
\end{proof}

Theorem \ref{theo:sep} provides a complete answer to the optimality of $\CH^1$.
The condition consists of one inequality per group.
From the proof, it is clear that reducing the number of inequalities is not possible.
However, if we relax the requirement of necessity, and consider a pair of links, there is a simpler sufficient condition that excludes the activation of both in any group.
This occurs, as formulated below, if the two links generate high interference to each other, but their rates are not much affected by simultaneous transmissions of the other links.

\begin{condition}
\label{cond:septwo}
For a pair of links $i, j \in \CN$, we define the following inequality.
$$
\frac{r_{i,\{i,j\}}}{r_{i, \CN \setminus \{j\} }} + \frac{r_{j,\{i,j\}}}{r_{j, {\CN \setminus \{i\}}}} \leq 1.
$$
\end{condition}

\begin{theorem}
\label{theo:septwo}
If Condition \ref{cond:septwo} is true, then there exists $\CH^*$ in which $i$ and $j$ do not appear together in any group, that is, in optimizing the schedule, the condition is sufficient for discarding all groups containing both $i$ and $j$.
\end{theorem}
\begin{proof}
Suppose an optimal schedule has a group $\Cc$ having both $i$ and $j$.
Without loss of generality, assume the time duration of $\Cc$ is 1.0.
The demands served equal $r_{i\Cc}$ and $r_{j\Cc}$ for the two links, respectively.
By the property of rate monotonicity we have, $r_{i\Cc} \leq r_{i, \{i,j\}}$ and $r_{j\Cc} \leq r_{i, \{i,j\}}$, which yields the following inequality
$$
\frac{r_{i \Cc}}{r_{i, \CN \setminus \{j\} }} + \frac{r_{j \Cc}}{r_{j, {\CN \setminus \{i\}}}} \leq 1.
$$
Consider, instead of $\Cc$, two groups $\Cc \setminus \{j\}$ and $\Cc \setminus \{i\}$. The rates of $i$ and $j$ are at least $r_{i, \CN \setminus \{j\}}$ and $r_{j, \CN \setminus \{i\}}$, respectively.
Activating the two groups for time durations $r_{i \Cc}/r_{i, \CN \setminus \{j\}}$ and $r_{j \Cc}/r_{j, \CN \setminus \{i\}}$ delivers respectively no less than $r_{i\Cc}$ and $r_{j\Cc}$ as served demands for $i$ and $j$, hence the conclusion.
\end{proof}

\begin{remark}
Both theorems \ref{theo:sep} and \ref{theo:septwo} significantly extend previous results of the optimality characterization of two links (see \cite{PaEp09} and the references therein).
In fact, for two links, $\CH^1$ is optimal if Condition \ref{cond:sep} holds for $\Cc = \{i,j\}$, otherwise $\CH^N$ is optimal.
\end{remark}

For $(N, \Bd, \Br)$, i.e., the symmetric case of cardinality-based rates, the number of inequalities in Condition \ref{cond:sep} is reduced to $N$.
This, together with the proof of Theorem \ref{theo:sep}, lead to the following corollary.

\begin{corollary}
\label{corr:tdma-cond}
The following condition is both sufficient and necessary for the optimality of $\CH^1$ for $(N, \Bd, \Br)$.
$$
m r_m \leq r_1~~m=2,\dots, N.
$$
\end{corollary}

The structure of Condition \ref{cond:septwo} also simplifies for $(N, \Bd, \Br)$.
In addition, by augmenting the line of arguments in the proof of Theorem \ref{theo:septwo}, we arrive at a sufficient condition for excluding the use of any group of a specific size $m$.
This fact is given in the corollary below.

\begin{corollary}
\label{corr:Hn-cond}
For $(N, \Bd, \Br)$ and a given group size $m \in [2, N]$, if the following condition holds for at least one $m'<m$, there is an optimal schedule not using any group of size $m$.
$$
m r_m \leq m' r_{m'}.
$$
\end{corollary}

\vspace{2mm}
By defining {\em sum-rate} as the amount of data served per time unit, Corollaries \ref{corr:tdma-cond} and \ref{corr:Hn-cond} have natural interpretations. The former corollary indicates that it is beneficial to use a TDMA-based schedule when any grouping of links results in lower sum-rate than single activation, while the latter one states that if there is a group size dominated in sum-rate by a smaller one, then there will be no optimal schedule selecting the larger size group.

Let us consider when it is preferable to augment the size of a group (of any size, except $N$).
Intuitively, one can expect that the group should be augmented with a new link, if the resulting sum-rate, is higher than that of any time combination of running the group and the link separately.
Conversely, if it is optimal to activate group $\Cc$, then the sum-rate of $\Cc$, namely $\sum_{i \in \Cc} r_{i \Cc}$, can not be achieved by any combined use of its $|\Cc|$ subsets of size $|\Cc|-1$.
The insight leads to the following condition.

\begin{condition}
\label{cond:notsep}
Given group $\Cc$, let $n = |\Cc|$ and denote by $\Cc_{\breve 1}, \Cc_{\breve 2}, \dots, \Cc_{\breve n}$ its $n$ subsets of cardinality $n-1$, obtained by deleting each of the $n$ links of $\Cc$.
Denote by $\Br_{\Cc} \in  {\mathbb R}_+^{n}$ the vector of rates of the links in $\Cc$, and $\Br_{\breve i} \in {\mathbb R}_+^{n}$ the corresponding rate vector for $\Cc_{\breve i}$ (with zero rate for $i$).
We define the following condition: For any $\Blambda = (\lambda_1, \dots, \lambda_{n})' \in {\mathbb R}_+^{n}$ with $\Be' \Blambda = 1$, the vector inequallity $\sum_{i \in \Cc} \lambda_i \Br_{\breve i} \leq \Br_{\Cc}$ is satisfied for at least one element.
\end{condition}

\begin{remark}
Note that finding whether or not there exists a $\Blambda$ vector that violates the condition can be formulated as an LP of size $O(n)$.
Thus the condition can be checked efficiently for any given group.
\end{remark}

What the above condition states is, in fact, that the rate vector of $\Cc$ can not be outperformed by the throughput region of the $n$ sub-groups.
If group $\Cc$ is active at optimum, then the condition must be true, as formulated below.

\begin{theorem}
\label{theo:notsep}
If $\Cc \in \CH^*$, then Condition \ref{cond:notsep} holds.
\end{theorem}
\begin{proof}
Suppose group $\Cc$ is activated with any positive time $T$.
Strictly inequality $\sum_{i \in \Cc} \lambda_i \Br_{\breve i} > \Br_{\Cc}$ in all the $n$ elements means that running $\Cc_{\breve 1}, \Cc_{\breve 2}, \dots, \Cc_{\breve n}$, with time proportions $\lambda_i, i=1,\dots,n$, respectively, will serve demand $T \Br_{\Cc}$ within less time than $T$, and the result follows.
\end{proof}

We now turn our attention to the scheduling strategy $\CH^N$.
In this solution, the $N$ groups, which are easily identified, are of sizes $N, N-1, \dots, 1$.
To save notation without loss of generality, assume link $1$ has its queue emptied first, followed by link $2$ in the second group, and so on.
Applying Theorem \ref{theo:notsep} yields immediately the following necessary condition for the optimality
of $\CH^N$.

\begin{corollary}
\label{eq:notsepequal}
If $\CH^N$ is optimal, then Condition \ref{cond:notsep} must be true for the $N-1$ groups $\{1,\dots, N\}, \{2,\dots, N\}, \dots$, and $\{N-1, N\}$.
\end{corollary}

Consider the implication of Condition \ref{cond:notsep} for cardinality-based rates $(N, \Bd, \Br)$.
Because of the rate symmetry, the quantity $\sum_{i \in \Cc} \lambda_i \Br_{\breve i}$ can attain maximum simultaneously in all the $n$ elements, only if $\lambda_i = \frac{1}{n}$ for all $i=1,\dots, n$. For this $\Blambda$,
all elements of $\sum_{i \in \Cc} \lambda_i \Br_{\breve i}$ equal
$\frac{n-1}{n} r_{n-1}$, resulting in the observation below.

\begin{corollary}
If $\CH^N$ is optimal for $(N, \Bd, \Br)$, then the following condition must hold.
$$
(m-1)r_{m-1} \leq mr_m~~ m =2, \dots, N.
$$
\end{corollary}

The inequalities in the above corollary form a hierarchy of relations with a clean interpretation. Namely, if $\CH^N$ is optimal, then the group sum-rate must be monotonically increasing in group size.
Conversely, if this monotonicity is violated, we conclude $\CH^N$ is not optimal.
However, the reverse formulation does not hold, i.e., the hierarchy of relations is not sufficient for ensuring that $\CH^N$ is optimal.
A counter-example is provided in Section \ref{sec:algoritm-opt}.

To arrive at a sufficient optimality condition for strategy $\CH^N$ for the general problem setting $(N, \Bd, F)$, we consider maximum and minimum rates of groups.
Denote by $r^{max}_m$ and $r^{min}_m$ the maximum and minimum link rates, respectively, of all groups of size $m$.
Although the exact values of them are difficult to calculate in general, in practical systems it is typically possible to derive optimistic respective conservative bounds on the rates, using the function $F$ and the channel matrix $\BG$, as replacements of the maximum and minimum values.

\begin{condition}
\label{cond:sufficientall}
We define the following $N-1$ inequalities,
$$
\frac{1}{r^{min}_{m}} + \frac{1}{r^{min}_{m-2}} \leq  \frac{2}{r^{max}_{m-1}}~~ m = 2,\dots, N.
$$
where, by convention, the term $\frac{1}{r_0^{min}}$, corresponding to $m=2$, is taken to be zero.
\end{condition}

The inequalities in Condition \ref{cond:sufficientall} form a chain for group sizes moving from one to $N$.
By the following theorem, this chain of relations is sufficient for the optimality of schedule $\CH^N$.

\begin{theorem}
\label{theo:sufficientall}
If Condition \ref{cond:sufficientall} holds, then $\CH^N$
is optimal.
\end{theorem}
\begin{proof}
Without any loss of generality, assume that, by strategy $\CH^N$, the queue of link $N$ becomes empty first, followed by link $N-1$, and so on.
Thus $\CH^N = \{\{1,\dots, N\}, \{1, \dots, N-1\}, \dots, \{1\}\}$.
In case of any degeneracy, the structure remains, with the only difference that the corresponding time  duration of some of these groups is zero.
All $T$-variables other than those for the $N$ groups are zeros.
Clearly, LP primal feasibility of \eqref{eq:lp} is satisfied by the solution.
Consider the LP dual \eqref{eq:dual}, and, for all the groups in $\CH^N$, set the corresponding rows in the dual to equality, that is,
\begin{subequations}
\label{eq:dualcond}
\begin{align}
\sum_{i=1}^N r_{i \CN} \pi_i & = 1, \label{eq:dualn} \\
\sum_{i=1}^{N-1} r_{i, \CN \setminus \{N\}} \pi_i & = 1, \label{eq:dualnminusone} \\
\sum_{i=1}^{N-2} r_{i, \CN \setminus \{N, N-1\}} \pi_i & = 1, \label{eq:dualnminustwo} \\
\dots & = 1, \nonumber \\
r_{11} \pi_1 & = 1. \label{eq:dualone}
\end{align}
\end{subequations}

The above $N$ equalities uniquely determine a solution to the LP dual \eqref{eq:dual}.
This, together with $\CH^N$, form a pair of dual and primal solutions.
Consider the complementary slackness condition in linear programming.
The condition for \eqref{eq:lpcons} is always satisfied no matter what the values of the dual variables are, since constraints \eqref{eq:lpcons} are equalities.
For the LP dual, complementary slackness holds for the above $N$ rows.
For the rest of rows, the condition is also satisfied because the corresponding $T$-variables are zeros.
In conclusion, the pair of solutions is optimal, if LP dual feasibility holds.

Suppose the derived dual solution is not dual feasible, i.e., at least one of the remaining constraints in \eqref{eq:dual} is violated.
We prove a contradiction, assuming a violated constraint concerning group $\{1, \dots, N-2, N\}$ of size $N-1$. The construction for arriving at a contradiction for other groups of size $N-1$ as well as other group sizes is analogous.

The above assumption of constraint violation means that $\sum_{i=1}^{N-2} r_{i, \CN \setminus \{N-1\}} \pi_i + r_{N, \CN \setminus \{N-1\}} \pi_N >1$. This implies the following inequality.
\begin{equation}
\label{eq:dualother}
r_{N-1}^{max} \sum_{i=1}^{N-2} \pi_i + r_{N-1}^{max} \pi_N > 1.
\end{equation}

Note that \eqref{eq:dualnminusone} gives the inequality
$r_{N-1}^{max} \sum_{i=1}^{N-1} \pi_i \geq 1$. This, together with
\eqref{eq:dualother}, result in

\begin{equation}
\label{eq:dualtemp}
r_{N-1}^{max} \sum_{i=1}^{N} \pi_i > 2 - r_{N-1}^{max} \sum_{i=1}^{N-2} \pi_i.
\end{equation}

Next, from \eqref{eq:dualnminustwo}, we obtain $\sum_{i=1}^{N-2} \pi_i \leq \frac{1}{r_{N-2}^{min}}$.
This observation and \eqref{eq:dualtemp} lead to the inequality below.

\begin{equation}
\label{eq:dualalmost}
\sum_{i=1}^N \pi_i > \frac{2}{r_{N-1}^{max}} - \frac{1}{r_{N-2}^{min}}.
\end{equation}

Scaling \eqref{eq:dualalmost} by $r_{N}^{min}$ along with applying
Condition \ref{cond:sufficientall} result in the following.

\begin{equation}
\label{eq:dualdone}
r_{N}^{min} \sum_{i=1}^N \pi_i > r_{N}^{min} (\frac{2}{r_{N-1}^{max}} - \frac{1}{r_{N-2}^{min}}) \geq 1.
\end{equation}

The strict inequality $r_{N}^{min} \sum_{i=1}^N \pi_i > 1$, however, contradicts \eqref{eq:dualn}, and the theorem follows.
\end{proof}

\vspace{2mm}
For $(N, \Bd, \Br)$, i.e. for the symmetric special case, the result of Theorem \ref{theo:sufficientall} can be strengthened.
We prove Condition \ref{cond:sufficientall} is also necessary for the optimality of $\CH^N$, as long as $\CH^N$ is non-degenerate, that is, all the $N$ groups are run with strictly positive time durations\footnote{For $(N, \Bd, \Br)$, non-degenerate $\CH^N$ is equivalent to all links having different demands.}.

\begin{theorem}
\label{theo:necessaryequal}
For $(N, \Bd, \Br)$ and non-degenerate $\CH^N$, Condition \ref{cond:sufficientall}, which reduces to the following inequalities, is also necessary for the optimality of $\CH^N$.
\begin{equation}
\label{eq:sufficientequal}
\frac{1}{r_{m}} + \frac{1}{r_{m-2}} \leq \frac{2}{r_{m-1}}~~ m = 2,\dots, N.
\end{equation}
\end{theorem}
\begin{proof}
As in the proof of Theorem \ref{theo:sufficientall}, assume without loss of generality that the sequence in which the queues empty in $\CH^N$ are $N, N-1, \dots, 1$.
Consider the group of size $m-1$, consisting of $\{1, 2, \dots, m-2, m\}$.
Note that the group is not part of the $\CH^N$ solution.
The base matrix $\BB$ of solution $\CH^N$ is triangular for $(N, \Bd, \Br)$, where column $k$
consists of $k$ consecutive elements of value $r_k$, followed by $N-k$ zeros.
Hence the basis inverse, $\BB^{-1}$, has the following form.

\[ \left( \begin{array}{cccccc}
\frac{1}{r_1} & -\frac{1}{r_2} & 0 &  \dots & 0 & 0\\
0 & \frac{1}{r_2} & -\frac{1}{r_3} & \dots & 0 & 0\\
\dots & \dots & \dots & \dots & \dots & \dots\\
0 & 0 & \dots & 0 & \frac{1}{r_{N-1}} & -\frac{1}{r_{N}} \\
0 & 0 & 0 & \dots & 0 & \frac{1}{r_{N-1}}
\end{array} \right)\]

For the aforementioned group, the linear programming reduced cost, for the basis $\BB$, is given by the expression below.
\begin{align*}
1 - (\frac{1}{r_N} - \frac{1}{r_{N-1}},  \frac{1}{r_{N-1}} - \frac{1}{r_{N-2}}, \dots, \frac{1}{r_1})
(0, \dots, 0, r_{m-1}, 0, r_{m-1}, \dots, r_{m-1})^T = ~& 2 - (\frac{r_{m-1}}{r_m} + \frac{r_{m-1}}{r_{m-2}})
\end{align*}

If the opposite of \eqref{eq:sufficientequal} holds, then the reduced cost is strictly negative.
Thus group $\{1, 2, \dots, m-2, m\}$ is an incoming variable in the simplex algorithm for linear programming.  As long as $\CH^N$ is non-degenerate, the pivot operation bringing in the group into the basis will strictly improve the objective function, and the theorem follows.
\end{proof}

It has been commented earlier that the condition of strictly improving throughput in group size, used in Corollary \ref{eq:notsepequal}, is not sufficient for the optimality of $(N, \Bd, \Br)$, whereas the inequalities given in Theorem \ref{theo:necessaryequal} are. Thus the former is implied by the latter
(in a strict sense, because they are not equivalent). This fact is formally established below.

\begin{corollary}
For $(N, \Bd, \Br)$, $\frac{1}{r_{m}} + \frac{1}{r_{m-2}}
\leq \frac{2}{r_{m-1}}, m = 2,\dots, N$, implies $(m-1)r_{m-1}
\leq mr_m, m =2, \dots, N$.
\end{corollary}
\begin{proof}
For $m=2$, the inequality gives $r_1 \leq 2r_2$, as $\frac{1}{r_0}$ is effectively zero by the aforementioned convention.
For the induction step, assume $(k-1)r_{k-1} \leq kr_k$ and consider $k+1$.
The inequality $\frac{1}{r_{k+1}} + \frac{1}{r_{k-1}} \leq \frac{2}{r_{k}}$ and the induction hypothesis together yield $\frac{1}{r_{k+1}} + \frac{k-1}{kr_{k}} \leq \frac{1}{r_{k+1}} + \frac{1}{r_{k-1}} \leq
\frac{2}{r_{k}}$.
Comparing the left and right sides, we obtain $kr_k \leq (k+1)r_{k+1}$, and the corollary follows.
\end{proof}

\begin{remark}
The inequalities in the conditions in this section do not involve $\Bd$.
Except from the non-degeneracy assumption in Theorem \ref{theo:necessaryequal}, all results are valid completely independent of the demand values.
This observation can be confirmed from \eqref{eq:lp}: Given a feasible schedule in form of a (non-degenerate) basic solution, whether or not it is optimal depends only on the left-hand side, which does not contain $\Bd$.
\end{remark}

\section{Complexity of Scheduling with Cardinality-Based Rates}
\label{sec:cardinality}

From Section \ref{sec:condition}, one can observe that the optimality conditions for $\CH^1$ and $\CH^N$ are more structured and stronger for $(N, \Bd, \Br)$.
This raises the question whether or not reaching optimality of $(N, \Bd, \Br)$ is more tractable than the general case.
In this section, we provide a positive answer to the question.

Consider first a more restrictive case, where the demand values are uniform.
For this setting, we provide an analytic solution requiring only linear time to compute, and prove it is globally optimal.

\begin{theorem}
\label{theo:polynomialequal}
For $(N, \Bd, \Br)$, let $m^* = \text{argmax}_{m=1}^N mr_m$.
If all demand values are uniform and equal to $d$, then the $N$ groups, $\{1, 2, \dots, m^*\}$, $\{2, 3, \dots, m^*+1\}$, $\dots$ $\{N, 1, \dots, m^*-1\}$, each scheduled for a time duration of $\frac{d}{m^*r_{m^*}}$, is optimal.
\end{theorem}
\begin{proof}
For all the links, the given schedule clearly meets demand $d$ exactly.
For any feasible scheduling solution (not restricted to the case in question) of length $T$, the total demand, $\sum_{i \in \CN} d_i$, divided by $T$, gives the average throughput per time unit.
As $\sum_{i \in \CN} d_i$ is a constant, a schedule is minimum in time if $(\sum_{i \in \CN} d_i)/T$ attains the maximum possible value.
By the assumption in the theorem, the instantaneous throughput of any feasible schedule can never exceed $m^*r_{m^*}$.
This throughput is achieved during the entire duration of the scheme in the theorem, and the result follows.
\end{proof}

\begin{remark}
Theorem \ref{theo:polynomialequal} generalizes a result in \cite{BoEp06}
concerning the much more restrictive case of $F_{\mathbb B}$, where $m^*$ corresponds to the size of the largest feasible matching.
In \cite{BoEp06}, however, all matchings of size $m^*$ are used for constructing the optimal schedule.
In our analysis, only $N$ groups are needed.
\end{remark}

For non-uniform demand, $(N, \Bd, \Br)$ does not admit an optimal schedule in closed form, yet we are able to conclude its polynomial-time tractability.
This fundamental insight is established in the following theorem.

\begin{theorem}
\label{theo:polynomial}
$(N, \Bd, \Br)$ is in class P, that is, the global optimum of any instance can be computed in polynomial time.
\end{theorem}
\begin{proof}
Consider the LP dual given in \eqref{eq:dual}.
For problem class $(N, \Bd, \Br)$, the dual has the following form.
\begin{subequations}
\label{eq:dualc}
\begin{align}
\max~~ & \sum_{i \in \CN} d_i \pi_i \label{eq:dualcobj}\\
\text{s.~t.}~~ & r_{|\Cc|} \sum_{i \in \Cc} \pi_i \leq 1~~\Cc \in \CH \label{eq:dualccons} \\
& \Bpi \geq 0
\end{align}
\end{subequations}

Observe that there is a symmetry among the occurrences of the dual variables in \eqref{eq:dualccons}.
As a result, given any feasible solution, swapping the values of any two dual variables will preserve feasibility.
The demand vector $\Bd$ is given in ascending order.
It follows that there must exist an optimal solution with $\pi_1 \leq \pi_2 \leq \dots \leq \pi_N$, because otherwise the objective function value can be improved or kept the same by swapping the variable values so that the condition holds.

Based on the above observation, one concludes that, among all constraints of \eqref{eq:dualccons} with $m$ variables on the left-hand side, the inequality $r_m \sum_{i={N+1-m}}^N \pi_i \leq 1$ is the most stringent  one in defining the optimum.
Therefore, the number of constraints required to define optimum can be reduced from $2^N-1$ to $N$, implying that, at optimum, the scheduling problem is equivalent to the following LP.

\begin{subequations}
\label{eq:dualcs}
\begin{align}
\max~~ & \sum_{i \in \CN} d_i \pi_i \label{eq:dualcsobj}\\
\text{s.~t.}~~ & \sum_{i=N+1-m}^N r_{m} \pi_i \leq 1, ~~ m=1, 2, \dots, N\label{eq:dualcscons}, \\
& 0 \leq \pi_1 \leq \pi_2 \leq \dots \leq \pi_N.
\end{align}
\end{subequations}

In conclusion, the optimal solution to problem class $(N, \Bd, \Br)$ is found by solving an LP of size $O(N)$, and the theorem follows.
\end{proof}

\vspace{1mm}
The above theorem is significant not only for the special symmetric case $(N, \Bd, \Br)$, but also for all scenarios where the transmitters having similar distances (and hence close-to-uniform channel gains) to their receivers, and the latter are located close to each other.
For these cases, one can expect that solving $(N, \Bd, \Br)$, which can be done fast, will provide a good
approximate solution to the global optimum.

\section{An Algorithmic Framework}
\label{sec:algorithm}

Since optimal scheduling is in general complex, it is important to propose algorithms which trade optimality against reduced complexity and yield decent performance.
To this end we propose several algorithm variations that range from suboptimal ones with low complexity to one actually optimal with high complexity.
They are all based on a common framework that uses a natural view of the problem and that is based to some degree on some of the optimality conditions and insights derived in the previous sections.
In fact we will demonstrate that the modular structure we propose eventually leads to exploiting tools from optimization theory so that we may come close to, or achieve full optimality, at a reduced complexity level.

As is evident from the LP formulation of the problem, any scheduling algorithm will have to have two basic components:\\
(i) a method for generating the $N$ link groups that will be part of the proposed final schedule, and\\
(ii) a method for deciding the duration of activation for each of these sets.\\
Later, we will confirm that this structural decomposition actually leads to a powerful toolset for eventual optimization.

Our proposed algorithms use a variety of criteria for fulfilling the two aforementioned requirements.
Each algorithm uses what we call a {\em Group Generation Module} to select the activation sets and an {\em Activation Duration Module} to decide the length of the activation of each set.
The two modules do not operate independently, but are closely coupled and operate interactively.
Before proceeding to the description and evaluation of these algorithms, we should emphasize that they are not based on ideas like those that govern the so-called approximation algorithms (e.g. \cite{Wan11}), or algorithms that impose structural restriction or additional assumptions on the problem.
Instead, they are completely generally applicable.
In fact we will show that some of these algorithms do achieve the optimal solution when some of the conditions that were mentioned earlier hold (i.e. in the cases where either $\CH^1$, or $\CH^N$ is optimal).

We proceed now with the description of the algorithms.
Regarding the Activation Duration module, we consider two possibilities. Either we activate the chosen group until one of its links empties its queue or we activate it for a fixed amount of time $\Delta$ chosen a priori as a parameter.
Clearly, in the latter case, it is possible that the time that the a queue empties is less that $\Delta$.
In that case the termination of the activation period occurs at that time instant, rather than continuing on until time $\Delta$ has passed.
Therefore for large values of $\Delta$ the two criteria become less and less distinguishable.
We refer to the first criterion as {\em TF} (for ``time at which the first queue of the group empties'') and to the second as {\em T$\Delta$} (for ``time $\Delta$, unless a queue empties earlier'').

Regarding the group generation module, some care needs to be exercised because here is the main source of high complexity.
Namely, there are $2^N-1$ possible groups. So to chose groups we must utilize some heuristic in the selection or, without any knowledge of the rate function $F$, we must endure the full consideration of all groups.
In either case, we want to reduce the complexity by avoiding the solution of the full LP in \eqref{eq:lp}.
Therefore we must choose a metric by which we will evaluate the candidate groups.
To this effect we either consider the sum-rate metric ({\em SR}), i.e. the quantity $\sum_{i \in \Cc} r_{i \Cc}$, or the weighted sum-rate ({\em WSR}), i.e. the quantity $\sum_{i \in \Cc} q_ir_{i \Cc}$, where $q_i$ is the ``current'' queue size at the transmitter of link $i$.
Clearly, at the start $q_i = d_i$; however, as different links get activated at different times, the initial $d_i$ keeps diminishing until it reaches zero.
Whether we choose the {\em SR} or the {\em WSR} metric, we have two choices for selecting a group.
Either we look at all $2^N-1$ groups (or all remaining groups, after some links have emptied) or we look at a judiciously chosen group that requires a much reduced search.
In the first case, we call the selection method {\em exact}, while in the second we call it {\em heuristic}.
In fact, we will later see that the ``{\em exact}'' choice can be achieved without necessarily looking at all $2^N-1$ possible groups.
Thus we have four possible group generation methods: (i) {\em SR--exact}, (ii) {\em SR--heuristic}, (iii) {\em WSR--exact} and (iv) {\em WSR--heuristic}. Since we may pair any one of these four group generation methods with either the {\em TF} criterion or the {\em T$\Delta$} criterion in the activation module we obtain a total of eight algorithms.

It remains to describe the heuristic method of choosing a group.
We propose the following.
We rank the $N$ links according to their currently remaining queue size $q_i$ in descending order.
To form a group we start with the singleton having the link at the top of the rank.
We then visit the second link of the ranking and pair it with the first one.
Doing so, the rates of concurrent activation of both links are reduced.
If the updated metric ({\em SR} or {\em WSR}) increases as a result of pairing the two links into a group, we keep the second link in the group.
Otherwise, we skip it.
We then visit the third link in the ranking and repeat the same process.
We proceed in this fashion until all links are visited.
Thus one group will emerge at the end of this process.
To hedge against this process being highly suboptimal, we repeat this entire construction for two additional rank-permutations, where we start with the second and third link in the ranking respectively.
Thus, in the end, we have three possible candidate groups from which we select the one with the highest metric.

Therefore we now have the following algorithms:
(1) {\em TF--SR--exact},
(2) {\em TF--SR--heuristic},
(3) {\em TF--WSR--exact},
(4) {\em TF--WSR--heuristic},
(5) {\em T$\Delta$--SR--exact},
(6) {\em T$\Delta$--SR--heuristic},
(7) {\em T$\Delta$--WSR--exact},
(8) {\em T$\Delta$--WSR--heuristic}.

\section{Generalizing the Framework with Optimization Tools}
\label{sec:general}

As we hinted earlier, the simple algorithms that were described in the preceding section can be embedded in a considerably more general setting that can exploit a variety of optimization techniques to yield a better combination of performance and complexity.
Specifically, we may now consider the Activation Duration module as a more sophisticated process.
Instead of choosing a simplistic criterion for activation (like the {\em TF} and the {\em T$\Delta$} ), it can actually obtain a ``tentative'' set of activation times that are actually optimal for a much reduced set of groups.
That is, it can be thought of as solving the LP over a small, limited, and restricted set of link groups.
Once it does this, it feeds back to the group generation module a ``metric'' that is based on the dual variables of the LP.
This metric is then used (in lieu of the {\em SR}, or the {\em WSR} metrics) to select a new group.
The new group is fed to the Activation Duration module that proceeds to resolve the LP and obtain a new tentative set of activation times.
It is possible that this new set improves on the previous one (that is, it does yield a shorter schedule length or it does not).
If it does, the process is repeated until no more improvement is achieved and the final activation times are then the ones that are obtained by the last iteration of the process.
This is the so-called Column-Generation method.

The remaining issue is how to select the new group to be added to the partial LP solution at each step of the iteration.
The dual variable values scaled by the rates are used as the metric in the group generation module.
These are two possibilities:
Either a heuristic can be used (like the one proposed for the four of the eight algorithms of the previous section, with the only difference that now the the dual variable $\pi_i$ is used to sort the links), or an exact determination of the next ``optimal'' group based on the dual variable metric.
If the latter option is chosen, then there are two further possibilities: either an exhaustive search performed over all remaining groups or a more efficient determination.
For the use of the second possibility one needs to know the rate function $F$ introduced in Section III.
Then, a variety of optimal (yet efficient) searches, from techniques of convex optimization, to branch-and-cut or branch and bound and other methods \cite{Bertsimas97, Bertsekas99, Chen10}, can be performed to determine the best group capitalizing on the knowledge and properties of the function.
If on the other hand only the rate values are known (but not the rate function that produced them) then the only option for exact determination of the best group at each step is the exhaustive search.

Thus, in the arsenal of the eight previously described algorithms we add two more.
The first uses the Column Generation method along with the rank-based heuristic for ``next group'' selection at each step (we call this the {\em CG--heuristic} algorithm).
The second uses the column generation method as well, but with an ``exact'' selection of the next group at each step.
We use the exhaustive search method for that while we note the possibility of dramatic expansion of the problem sizes (in terms of number of links $N$) that are possible to solve if we utilize the knowledge of the rate function (we call that the {\em CG--exact} algorithm).

Note also that even in the earlier eight algorithms the ``exact'' group generation option can be exercised with significantly reduced complexity if the rate function is known.
The only difference there is that the metric for the candidate groups is not the one that depends on the dual variables, but rather the {\em SR} or {\em WSR}, or possibly even a totally different metric.
Of course the use of other metrics, either in association with the LP or the {\em TF}, or {\em T$\Delta$} methods does not in general lead to the minimum length schedule, while in the case of the {\em CG--exact} algorithm it generally does.
The complexity of the {\em CG--exact} algorithm though is in the worst case exponential.

\section{Optimality Properties of the Algorithms}
\label{sec:algoritm-opt}

We present now some properties of the proposed algorithms that strengthen their appropriateness for the solution to the scheduling problem.
%
%
%
In the following, we denote an algorithm through the descriptive terms used in the preceding sections.
For example, $\langle T\Delta, SR \rangle$ denotes the strategy of selecting the group with maximum sum-rate in every iteration, and the time duration of activation is a constant $\Delta > 0$ (or until the first queue in the group empties, whichever occurs first).
It is clear that each of the five designs of the Activation Duration module will empty all queues in a finite number of iterations.
Here, one iteration of the module refers to the activation of one group under {\em TF} or {\em T${\Delta}$}, or solving a whole LP under {\em CG}.
Obviously, the running time of one iteration is polynomial in $N$ in all cases.
As for the number of iterations to empty all queues, the complexity is summarized in the following theorem.

\begin{theorem}
\label{theo:complexityiteration}
1) The number of iterations under {\em CG--exact} is polynomial in $N$,
2) the number of iterations under $\langle TF, SR \rangle$ and $\langle TF, WSR \rangle$ is $O(N)$,
and
3) the number of iterations by $\langle T\Delta, SR\rangle$ and $\langle T\Delta, WSR \rangle$ is $O(N \frac{\max_{i \in \CN}d_i }{\Delta r^*})$ and hence pseudo-polynomial in $N$, where $r^*$ is either
a) $\min_{i \in \CN}{r_{i\CN}}$, for the case of a continuous rate function $F_\Cc$, or
b) the lowest positive rate level of a discrete rate function $F$.
\end{theorem}

\begin{proof}
An LP can be solved to optimality using a polynomial number of iterations, or equivalently, separation of constraints in the dual LP, therefore the first statement holds.
The second statement follows from the fact that at least one link gets its
queue emptied in every iteration in $\langle TF, SR \rangle$ and $\langle TF, SR \rangle$.
The last statement follows from $\sum_{i \in \CN} d_i \leq N \max_{i
\in \CN} d_i$, and that $\langle T{\Delta}, SR\rangle$ and $\langle T\Delta, WSR \rangle$ drain at least an amount proportional to $\Delta r^*$ from one or several queues per iteration.
\end{proof}

Among the above designs {\em CG--exact} is an exact algorithm that guarantees global optimality \cite{BjVaYu04}).
The correspondence to the general method of Column Generation consists of the fact that a column in the LP \eqref{eq:lp} corresponds a variable associated with a group.
Note that the complexity of {\em CG--exact} by Theorem \ref{theo:complexityiteration} does not contradict the general ${\cal NP}$-hardness of the scheduling problem.

Let us consider the other four options of section \ref{sec:algorithm}.
In general, they are sub-optimal, but significantly simpler than {\em CG}.
In the following, we show that, if the corresponding sufficient optimality conditions discussed in Section
\ref{sec:condition} strongly hold (i.e., the inequalities in the conditions are strict), the use of {\em SR} gives the optimal schedule, i.e., the two base scheduling solutions $\CH^1$ and $\CH^N$, respectively.

\begin{theorem}
\label{theo:performancerate}
Both $\langle TF, SR \rangle$ and $\langle T{\Delta}, SR\rangle$  with any $\Delta > 0$ lead to schedules $\CH^1$ and $\CH^N$, respectively if, under exact group selection, Conditions \ref{cond:sep} and \ref{cond:sufficientall} strictly hold.
\end{theorem}
\begin{proof}
Assume Condition \ref{cond:sep} is strictly satisfied and, without loss of generality, that the rates are in descending order for individual links, that is, $r_{11} \geq r_{22} \geq \dots \geq r_{NN}$.
For any group $\Cc$ with $|\Cc| \geq 2$, we have

$$
\sum_{i \in \Cc} \frac{r_{i\Cc}}{r_{ii}} < 1 \Leftrightarrow
\sum_{i \in \Cc} r_{i\Cc} \frac{r_{11}}{r_{ii}} < r_{11}
\Rightarrow \sum_{i \in \Cc} r_{i \Cc} < r_{11}.
$$

The last inequality above follows from $\frac{r_{11}}{r_{ii}} \geq 1, i \in \Cc$.
As a result, $\{1\}$ is the group to be activated.
For {\em TF}, the demand of link one is served for time duration $d_1/r_{11}$, and the next group to be  activated with the {\em SR} metric is $\{2\}$, and so on.
Applying {\em T$\Delta$}, the demand of link one is gradually served by repeatedly activating $\{1\}$, after which $\{2\}$ is used.
Hence both strategies lead to schedule $\CH^1$.

If we assume Condition 4 strongly holds, we first prove that $(m-1)r_{m-1}^{max} < m r_{m}^{min}, m=2, \dots, N$.  For $m=2$, $r_1^{max} < 2 r_{2}^{min}$ follows immediately from the assumption.
Suppose $(k-1)r_{k-1}^{max} < k r_{k}^{min}$ holds.
Then, Condition 4 written for $k+1$ leads to the following

$$
\frac{1}{r_{k+1}^{min}} +
\frac{k-1}{k r_{k}^{max}}
< \frac{2}{r_{k}^{max}}
$$

which immediately implies $k r_{k}^{max} < (k+1) r_{k+1}^{min}$.
As a result, the total sum-rate increases for any group by adding new links.
Therefore, the grand group $\CN$ has the highest sum-rate under {\em SR}. For both {\em TF} and {\em T$\Delta$}, the group will be activated until one of the links' queue becomes empty.
Repeating the above argument, the next group to be activated is the one consisting of all links with positive remaining demands, and the theorem follows.
\end{proof}

The construction of the proof leads to another observation of $\langle
TF, SR \rangle$ and $\langle T{\Delta}, SR\rangle$. Under the exact,
or a deterministic heuristic for group selection with the {\em SR}
metric, {\em T$\Delta$ } will be activating the same group until one
of its links' queue empties, becoming equivalent to {\em TF}.

Theorem \ref{theo:performancerate} supports the choice of {\em SR} as the group selection metric.
Below, we illustrate the merit of the {\em WSR} and that of {\em T${\Delta}$} using the following two examples.

\begin{example}
\label{ex:nothighest}
Consider a case where $N=3$, $ d_1=d,~ d_2=2d,~d_3=3d$, and $r_1 = 6,~r_2 = 5, r_3 = 4,$, i.e. $(3, (d, 2d, 3d), (6, 5, 4))$.
The unique optimum comprises the two groups $\{1, 3\}$ and $\{2, 3\}$, with time durations $\frac{d}{5}$ and $\frac{2d}{5}$, respectively.
\end{example}

Note that for this example $3r_3 > 2r_2 > r_1$, yet group $\{1,2,3\}$ having the highest sum-rate is not part of the optimum.
Hence preferring the top sum-rate group (which itself is not a trivial task) in designing a schedule may not work well.


\begin{example}
\label{ex:share}
Consider $(3, (d, d, d), \Br)$, with $2r_2 > 3r_3$ and $2r_2 > r_1$.
One can easily verify that the unique optimum schedule consists
of groups $\{1,2\}$, $\{2,3\}$, and $\{1,3\}$, each with a time duration
of $\frac{d}{2r_2}$.
\end{example}

Observe that, for Example \ref{ex:share}, {\em none} of the links has
its entire queue emptied in any of the groups that it participates in at the
unique optimum. Thus the {\em TF}--based algorithm may fail, {\em even if an exhaustive search}
of all ordered combinations of $N$ groups is tried.

\begin{remark}
For Example 1, using the {\em WSR} metric enables the construction of an optimal schedule.
Specifically, $\langle TF, WSR \rangle$  yields an optimum, and and $\langle T{\Delta}, WSR\rangle$ is also optimal for $\Delta>\frac{4}{15}d$.
For Example 2, one can verify that $\langle T{\Delta}, SR\rangle$ delivers optimum for
all $\Delta =\frac{d}{2^kr_2}$, ($k$ is a positive integer),
and $\langle T{\Delta}, WSR\rangle$ is optimal for $\Delta = \frac{d}{2^ur_2}$,
where $u$ is positive integer that satisfies $u>\log_2 \frac{r_2}{2r_2-r_1}$ .
\end{remark}

\section{Simulation Setup and Results}
\label{sec:simulation}

In this section we provide simulation results to illustrate the performance of the algorithms developed within the framework.
We consider a set of $N = 15$ links randomly placed in an area of 1000$\times$1000 meters.
The signal propagation follows a distance-based model with a path loss exponent of 4.
The distance between the transmitter and the receiver of a link was restricted to be between 3 and 250 meters to obtain links of practically meaningful SNR values.
For the queue sizes, we defined two different sets:
(i) uniform demand of 1000 bits,
(ii) non-uniform demand, uniformly distributed in [100 \dots 1500] bits.
In each setup, 100 link location instances were ran, unless otherwise indicated.

For the rate values, we consider two cases, namely
(i) rates given by the Shannon formula as in Eq. (\ref{eq:shannon}),
(ii) rates given by a combination of uncoded BPSK with symbol rate control, at a fixed error rate, as in \cite{PaEp08b}.
In the Appendix we provide a detailed derivation of the BPSK rate formula used, with the
standard assumption that the energy of interference from concurrent transmissions is equivalent to  Average White Gaussian Noise (AWGN).
For illustration purposes we provide both functions in Figure \ref{fig:Shannon-BSPK} of the Appendix.
Note that although both are approximate, with Shannon's formula giving an upper bound on the achievable rates, and BPSK giving a more practical flavor in our investigation, they jointly provide a useful insight on how the physical layer affects the algorithms' performance.

For each of the setups described above, we solved the full LP of
(\ref{eq:lp}), using AMPL \cite{AMPL02}, to establish the optimal
schedule length, assuming a unit of bandwidth in Hz and error rate $z
= 10^{-6}$ for the BPSK rate calculations.  Then, all ten algorithms
of the previous sections were ran for all instances.  All results
presented below are normalized with respect to the ''baseline''
optimal value.

In Figure \ref{fig:TS-analysis} we present the effect of the $\Delta$
parameter on the performance of the {\em T$\Delta$}-based algorithms.
Intuitively, a small $\Delta$ enables a more ``cautious'' design since
the algorithms will iterate more times between the two modules, thus
having more opportunities to select groups closer to the optimal with
more refinement. This is directly confirmed by the green ({\em
T$\Delta$--WSR--Heuristic}) and black ({\em T$\Delta$--WSR--Exact})
algorithms, for which the gap increases when $\Delta$ grows.  Note
that for a very large $\Delta$ value, the {\em T$\Delta$} strategy
would coincide with the {\em TF} one, as can be seen in the rightmost
of Figure \ref{fig:TS-analysis}.
In contrast to {\em T$\Delta$--WSR--Heuristic} and {\em
T$\Delta$--WSR--Exact}, the red line for algorithm {\em
T$\Delta$--SR--Heuristic} has the opposite trend, as small $\Delta$
gives larger optimality gap. This is due to the behavior of the
heuristic -- the heuristic sorts the links by their remaining demand
in group generation, even though the sorting may not respond well to
the {\em SR} metric. The mismatch is however rectified for larger
$\Delta$ values since the impact of the heuristic will be accumulated
fewer times. The blue line for algorithm {\em
T$\Delta$--SR--Exact} is horizontal, because, with exact group
selection and the {\em SR} metric, the maximum sum-rate group will
remain selected, regardless of the size of $\Delta$, until one of the
link queues empties. Thus T$\Delta$ is equivalent to {\em TF}, as
commented in the discussion after Theorem
\ref{theo:performancerate} in Section
\ref{sec:algoritm-opt}.

For {\em T$\Delta$--WSR--Heuristic} and {\em T$\Delta$--WSR--Exact},
the gain of reducing $\Delta$ diminishes below some point
($\sim$0.5s in Figure \ref{fig:TS-analysis}), since, without
significant queue draining, the same group is simply selected over and
over. For the minimum $\Delta$ used (i.e., leftmost of the figure),
the best performance is given by {\em T$\Delta$--WSR--Exact}, as also
suggested by the first example in Section \ref{sec:algoritm-opt}.  On
the rightmost (i.e., {\em TF} activation), the {\em WSR} metric is
inferior to the {\em SR} metric, because the former may result in a
group with low sum-rate, and the impact can not be mitigated by group
activation which runs the group until one link empties the entire
remaining queue. As a result, for large $\Delta$ or equivalently {\em
TF} activation, the best performance is achieved by {\em
T$\Delta$--SR--Exact}. Finally, the performance of heuristic selection
is consistently inferior to exact selection in Figure \ref{fig:TS-analysis}.

\begin{figure*}[tbp]
    \centering
    \subfigure[Equal demands]{
                \includegraphics[width=0.48\textwidth]{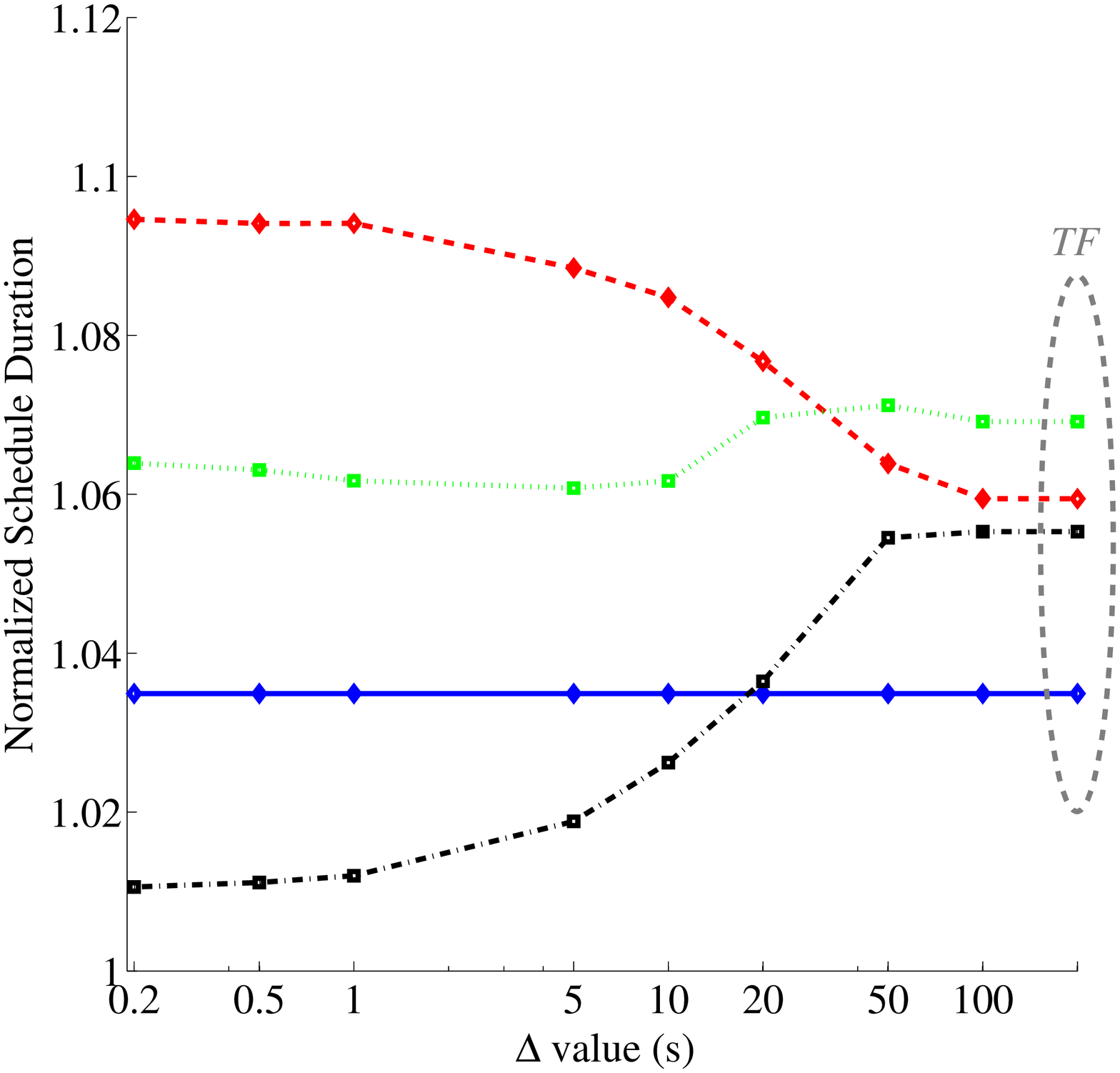}
                \label{fig:TS-analysis-ED}
                }
    \subfigure[Random demands]{
                \includegraphics[width=0.48\textwidth]{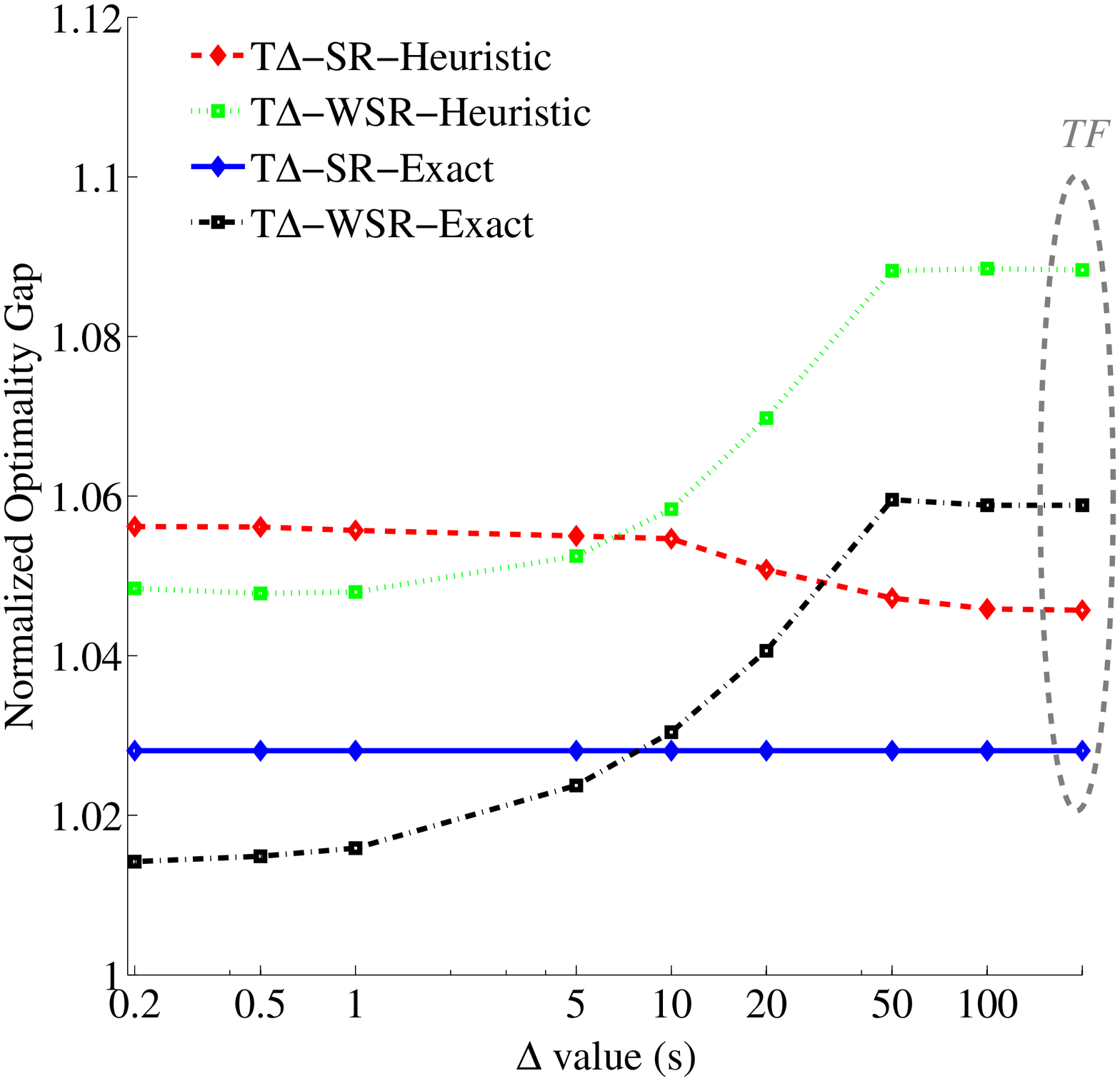}
                \label{fig:TS-analysis-RD}
                }
    \caption{Effect of activation time $\Delta$ on the schedule duration for the four {\em T$\Delta$}--based algorithms, under the Shannon rate function (x-axis in logarithmic scale; each point is an average value over 50 simulation instances).}
\label{fig:TS-analysis}
\end{figure*}

\begin{figure}[tbp]
\vspace{-5mm}
    \centering
    \subfigure[Equal demands]{
                \includegraphics[width=0.48\textwidth]{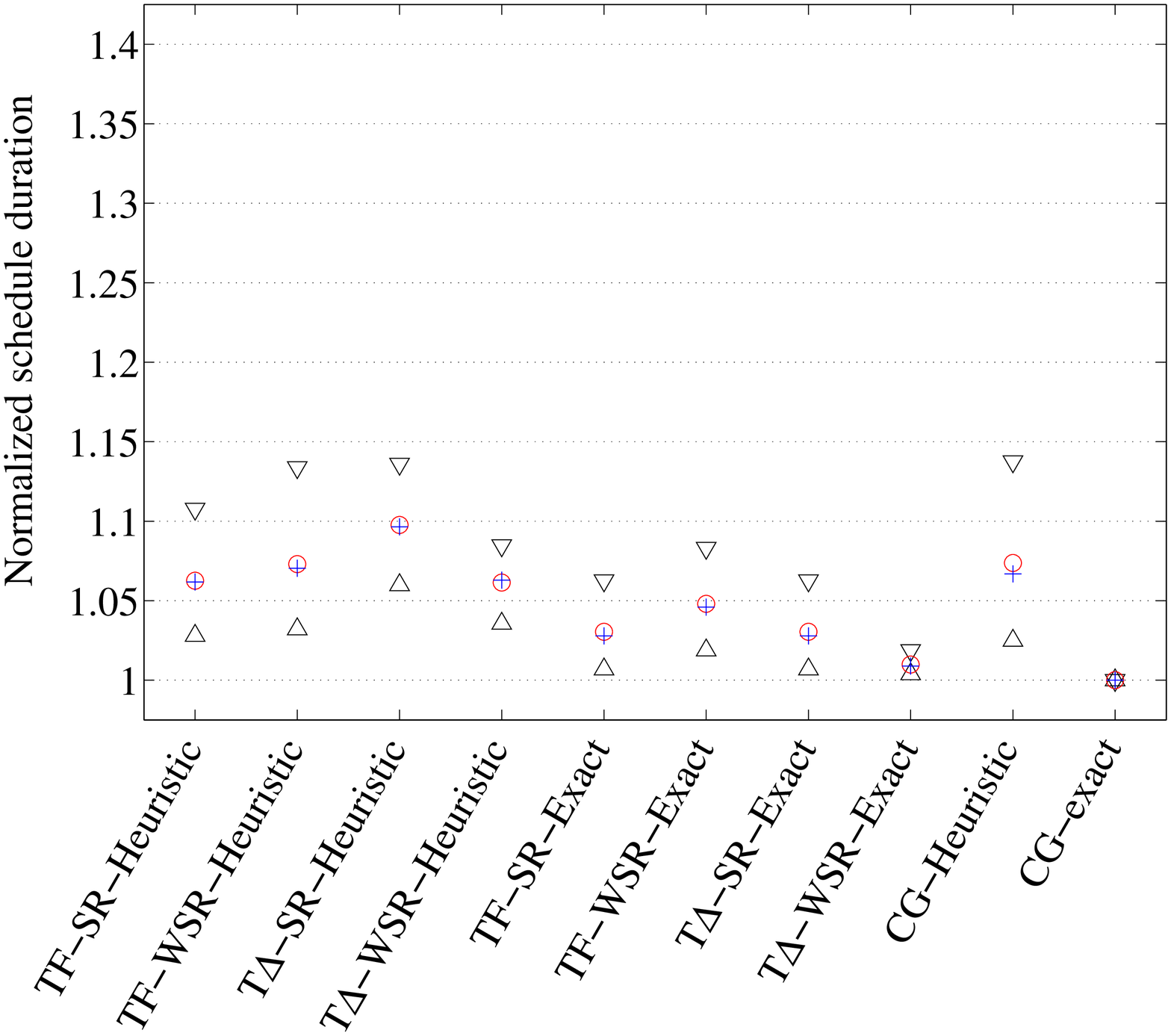}
                \label{fig:Shannon-U}
                }
    \subfigure[Random demands]{
                \includegraphics[width=0.48\textwidth]{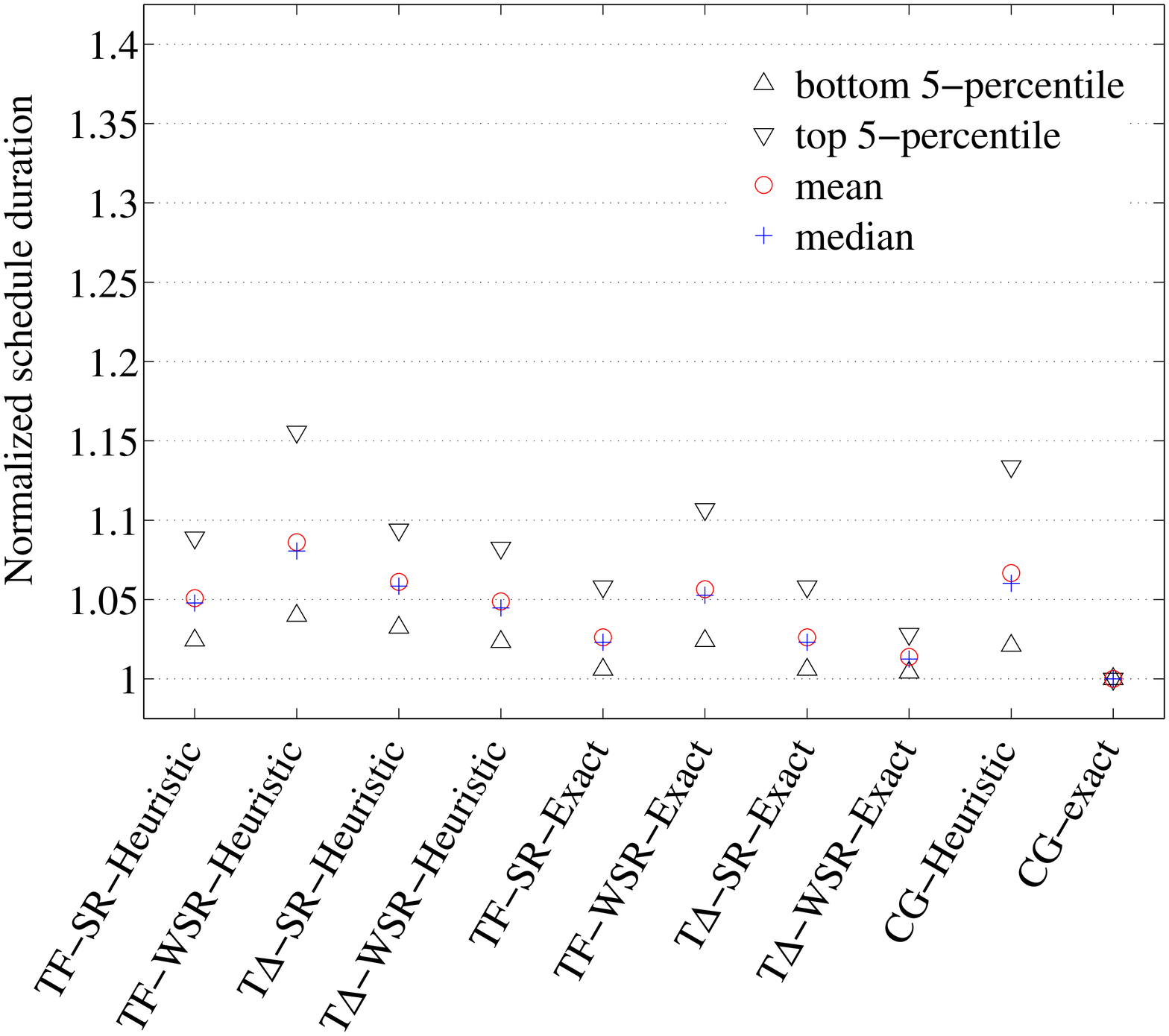}
                \label{fig:Shannon-NU}
                }
    \caption{Normalized schedule durations for all the algorithms under the Shannon-based rate function of Figure \ref{fig:Shannon-BSPK}.}
    \label{fig:Shannon}
\end{figure}

\begin{figure}[tbp]
\vspace{-5mm}
    \centering
    \subfigure[Equal demands]{
                \includegraphics[width=0.48\textwidth]{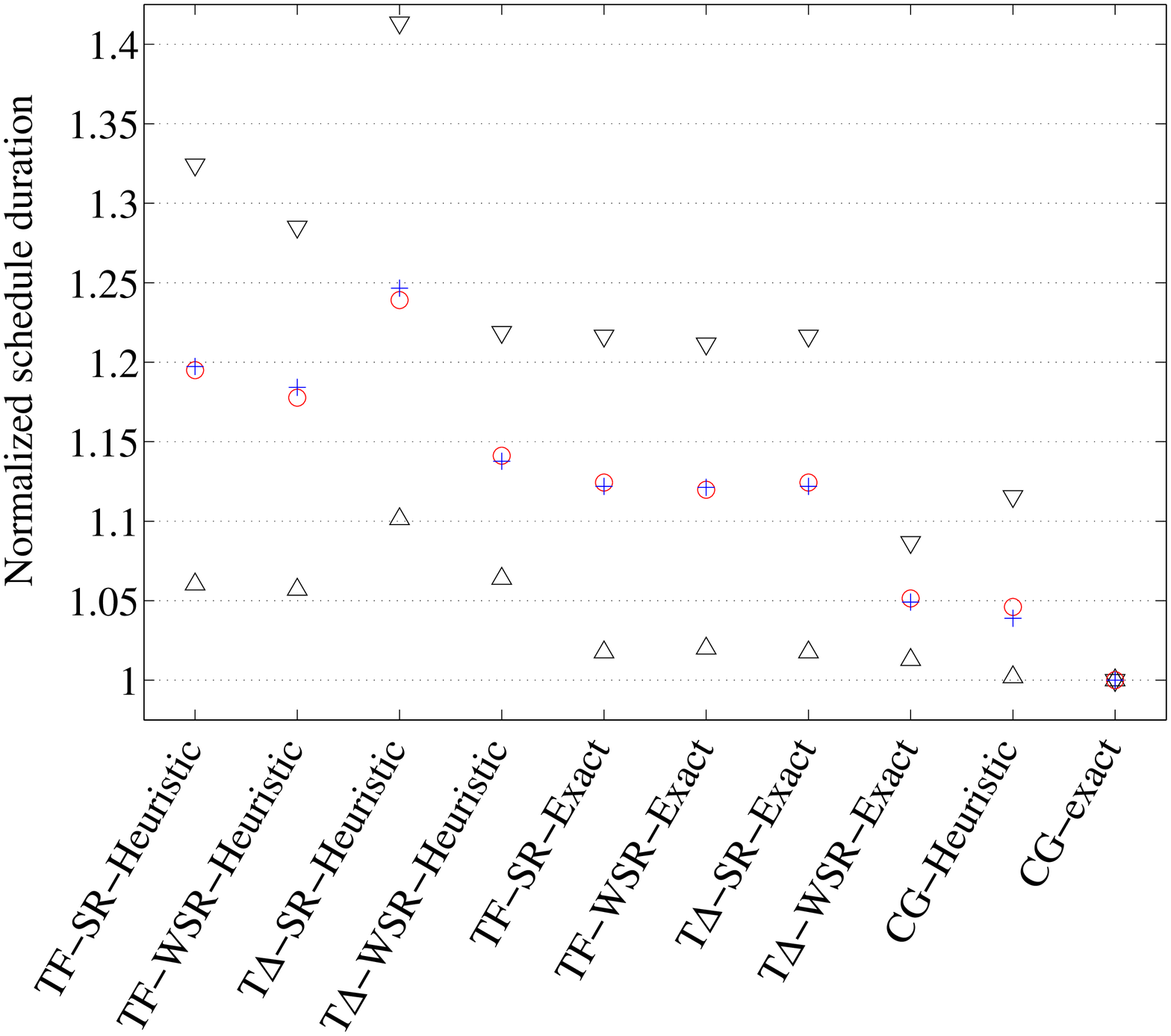}
                \label{fig:BPSK-U}
                }
    \subfigure[Random demands]{
                \includegraphics[width=0.48\textwidth]{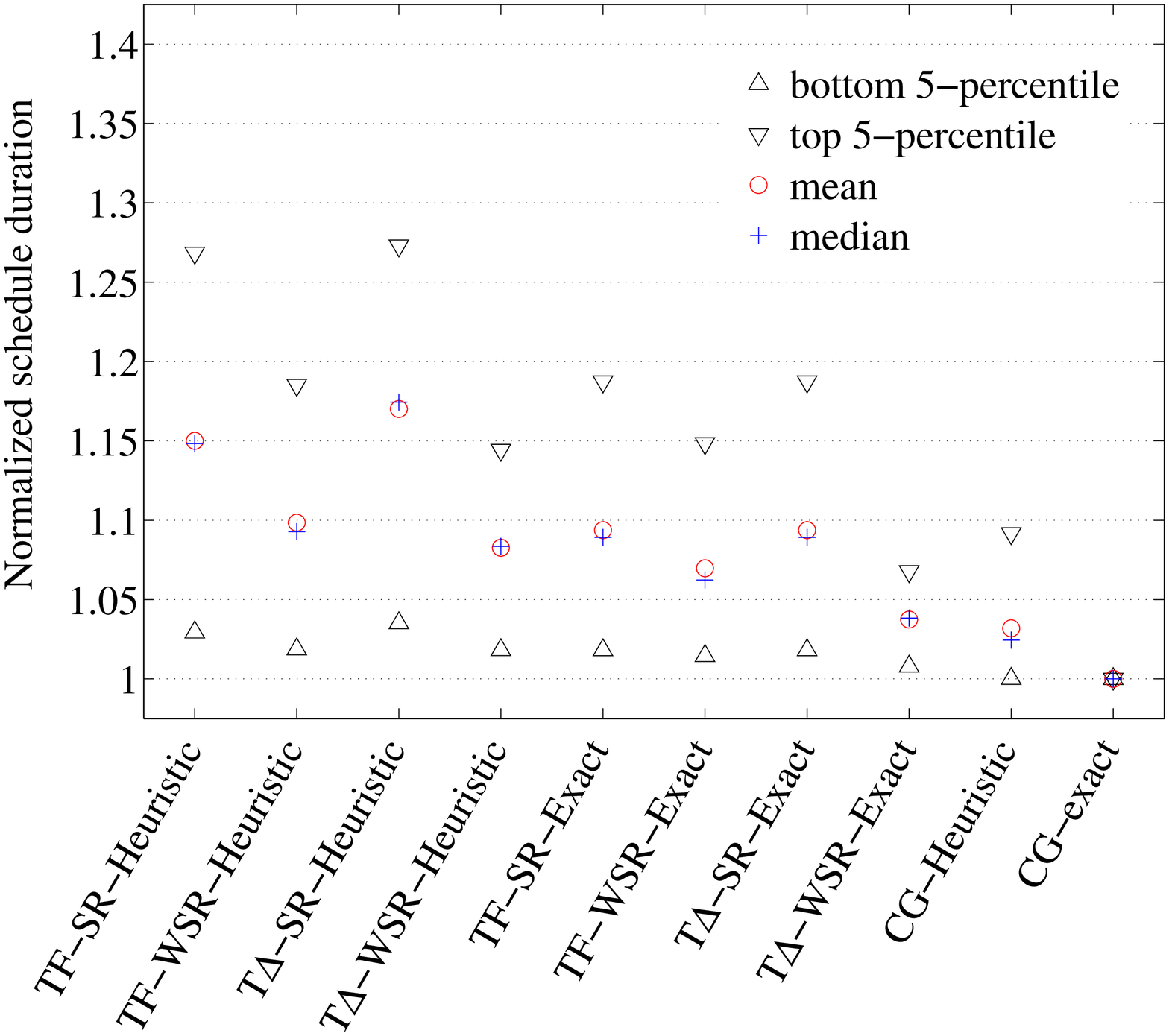}
                \label{fig:BPSK-NU}
                }
    \caption{Normalized schedule durations for all the algorithms under the BPSK-based rate function of Figure \ref{fig:Shannon-BSPK}.}
    \label{fig:BPSK}
\end{figure}

Figures \ref{fig:Shannon} and \ref{fig:BPSK} provide performance
comparison for all the algorithms.  For {\em T$\Delta$}-based
algorithms, based on our preceding discussion, a small value (0.5s) is
chosen for the activation duration parameter $\Delta$. The last
algorithm, {\em CG-exact}, always gives the global optimum,
normalized to 1.0 in the figures. Overall, the other nine algorithms
perform reasonable well, with only one giving an average optimality
gap larger than 20\%.

The first eight algorithms all perform better when the rates are
derived from the Shannon formula, in both demand cases. The
explanation lies in the shapes of the two rate functions. The BPSK
rate is must more robust to interference from concurrent transmissions
in comparison to the Shannon function. Hence, at optimum for the BPSK
rate function, large-cardinality groups having similar rates on the
links are very likely to be used. With the Shannon formula, smaller
groups of higher link rates are more preferred in the optimal
schedule. Consequently, scheduling with BPSK rate is much more prone to
sub-optimality in group selection (as many groups perform similarly
for both {\em SR} and {\em WSR}), and, more importantly, to
sub-optimality in group activation (cf. Example 2 in Section
\ref{sec:algoritm-opt}). Thus {\em TF} and {\em T$\Delta$} become more sub-optimal
for BPSK-derived rates than those derived by the Shannon formula.
This conclusion is further supported by the performance of {\em
CG--Heuristic}. In this case, groups activation is carries out group
activation with LP, which is the best possible solution of determining
the time share among groups, justifying the better performance of
this algorithm for the BPSK case, as opposite to the two heuristic
activation strategies.

For both {\em TF} and {\em T$\Delta$}, heuristic group selection is
always outperformed by the exact one. Hence enhancing the group
selection module alone gives noticeable contribution to the overall
performance, regardless of the activation strategy. As expected, the
impact of sub-optimality in group selection is more striking in BPSK.
As was mentioned above, larger groups are expected at optimum for
BPSK-derived rates.  When the exact solution of group selection
contains many elements, it is less probable that greedy selection, as
used in our heuristic, is able to approach optimality.

Comparing the two metrics {\em SR} and {\em WSR}, the latter always
yields better results for the {\em T$\Delta$} activation strategy, that
is, the notion of remaining demand interacts better with emptying
queues progressively. This was also observed in the discussion of
Figure \ref{fig:TS-analysis}. For BPSK, {\em WSR} also outperforms
{\em SR} for {\em TF} activation.  The reason is that the
optimal schedule with BPSK tends to use groups of similar sizes. With
{\em SR}, there is a higher risk that a small-cardinality group with
low sum-rate and high remaining demand will have to be deployed by the
end, making the overall schedule inferior in comparison to balancing the remaining queues in group
selection.  For Shannon-based rates, the structure at optimum is quite
the opposite, and {\em SR} behaves better than {\em WSR} in {\em TF}.

From the above two Figures, Examples 1 and 2, and the results in
Figure 1, it is inconclusive whether the {\em TF}, or {\em T$\Delta$}
activation should be in general preferred. However, if the {\em WSR}
metric is employed in group selection, {\em T$\Delta$} is clearly
superior, as discussed above.  In addition, the two group activation
strategies coincide in {\em TF--SR--Exact} and {\em
T$\Delta$--SR--Exact}, as justified by our discussion in Section
\ref{sec:algoritm-opt} and the results in Figure 1.

Finally, the demand structure (uniform versus non-uniform) has a
noticeable impact on performance, when heuristic group selection is
used. In general, the results show improved performance when the
initial demand is non-uniform. This can be attributed to the demand
ordering in group construction of the heuristic that we used.
Non-uniform demands aid the heuristic to better differentiate among
the links, especially in the BPSK case where the effect of non-uniform
demand is indeed more prominent, as the links in a group tend to have
more similar rates. With exact group selection coupled with {\em TF}
and {\em T$\Delta$}, the demand structure has virtually no effect on
performance for the Shannon-based rate function.  For BPSK-derived
rates, on the other hand, non-uniform demand leads to smaller
optimality gap. This is because the sub-optimality of {\em TF} and
{\em T$\Delta$} in group activation is more crucial for uniform demand
that assembles the structure of Example 2.

\section{Conclusion}
\label{sec:conclusion}

We have considered the minimum-length scheduling problem for the case of emptying $N$ queues over a shared channel.
The generic consideration of rates, which may be produced by some underlying function, unifies the previously considered formulations.
Several fundamental results of solution characterization have been gained.
First, we have proven the hardness of the problem for all continuous and monotonically increasing function in SINR.
Second, optimality conditions of two base scheduling strategies are developed and formalized. Third, we have demonstrated how the problem class with cardinality-based rates can be solved effectively. On the algorithmic side, we have presented a framework that accommodates both exact and sub-optimal scheduling solutions. Extensive simulation results have been provided and assessed to quantify the performance of some specific algorithm designs.

The research line of the current paper is subject to several
extensions.
For example, we may consider cooperative methods among the links, including relaying each other's messages. In that case we could have multiple transmitters that transmit to the same receiver draining the same queue simultaneously.
Another extension is the fundamental solution characterization under a multi-objective setting that incorporates both efficiency (i.e., schedule length) and energy expenditure, or that includes the aspect of fairness among the links.
Multi-hop or multicasting applications are also of interest.
Last but not least, it is important to use the insights from this work in the problem of scheduling with continuous arrivals (rather than the queue draining problem).

%

\appendix[The BPSK rate function for a multi-user environment]
\label{sec:bpsk}

Assuming an uncoded BPSK modulation scheme, in an interference-free environment the bit error probability $z$ is given by
$$
z = Q\Bigg(\sqrt{\frac{2E_b}{N_0}} \Bigg),
$$
where the function $Q(x)$ is the probability that a Gaussian random variable with zero mean and unit variance exceeds $x$. The $\frac{E_b}{N_0}$  fraction is the system SNR, where the numerator is the bit energy and the denominator is the power spectral density (psd) of the noise, both in Joules \cite{Andrea05}.

In the presence of interference by concurrent transmissions, denoting by $T_0$ the duration of one BPSK symbol the bit rate will be $r_b = \frac{1}{T_0}$. The received bit power is also $P_b = \frac{E_b}{T_0}$ in Watts. Then the error rate can be calculated as
$$
z = Q\Bigg(\sqrt{\frac{2E_b}{I_N}}\Bigg)=Q\Bigg(\sqrt{\frac{2P_bT_0}{I_N}}\Bigg)=Q\Bigg(\sqrt{\frac{2P_b}{I_Nr_b}}\Bigg),
$$
where $I_N$ denotes the sum of all interference energy in Joule (which we consider it to be an AWGN signal) plus the noise psd. Notice that $I_N$ can be approximated by scaling the sum of the interference powers received by a time factor, in order to obtain an appropriate quantity in Joules.

Observe now that any change in the $I_N$ value, under a fixed $z$ and fixed $P_b$, leaves us with the symbol duration $T_0$, i.e. the bit rate, as the only control we have to keep the equation above satisfied.
Hence, solving the error rate equation above yields our approximation to the BSPK bit rate:
$$
r_b = {\frac{2}{Q^{-1}(z)^{2}}}{\frac{P_b}{I_N}},
$$
where $Q^{-1}$ is the inverse $Q$-function.

Note that for general pulses the BSPK symbol rate must satisfy $r_s = B/k$, where $k$ is the spectral efficiency (which for ideal pulses we can assume to be equal to one) and $B$ is the channel bandwidth \cite{Andrea05}. Hence, we limit our BPSK rate function results to a maximum rate value $r_b^{max} = r_s = B$, i.e.
$$
r_b = min\Bigg\{{\frac{2}{Q^{-1}(z)^{2}}}{\frac{P_b}{I_N}}, B\Bigg\}.
$$


\begin{figure}[htbp]
\centering
\includegraphics[scale=0.56]{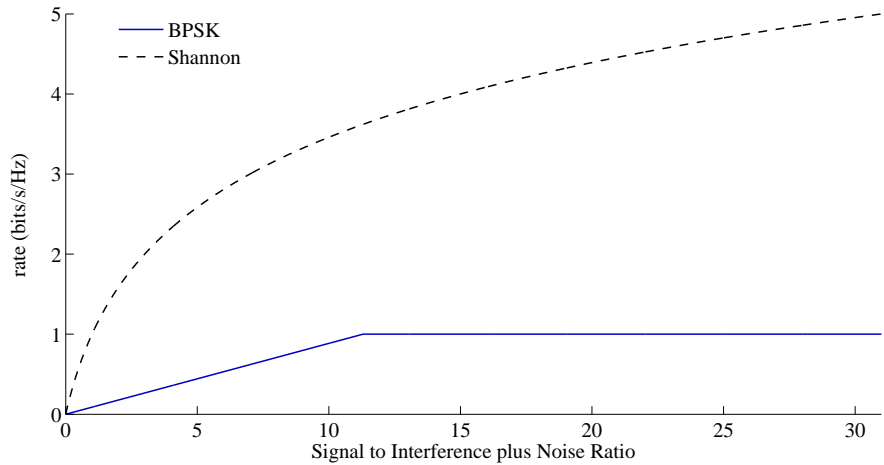}
\caption{The Shannon-based and BPSK-based rates, per bandwidth unit, versus the SINR.}
\label{fig:Shannon-BSPK}

\end{figure}

\end{document}